\documentclass{article}

\usepackage[colorlinks=true,breaklinks=true,bookmarks=true,urlcolor=blue,
     citecolor=blue,linkcolor=blue,bookmarksopen=false,draft=false]{hyperref}

\def\EMAIL#1{\href{mailto:#1}{#1}}% When hyperref is used, otherwise outcomment 
\usepackage{graphicx}  

\usepackage{amssymb}
\usepackage{amsmath}  
\usepackage{bbm}
\usepackage{graphicx} 

\usepackage{algorithm} 
\usepackage{algorithmicx}
\usepackage{algpseudocode} 
\usepackage{tabularx, booktabs}

\usepackage{enumerate}
\usepackage{color}
\usepackage{mathrsfs}

\def\spi{\bar{\Pi}} %submitted preference
\def\SPi{\bar{\Pi}} %submitted prefrence profile

\def\SP{\Pi}  % Sincere Preferences
\def\PP{\spi}             % Putative Preferences

\usepackage{adjustbox}
\usepackage{caption}
\usepackage{graphicx}
\usepackage{varwidth}

\usepackage{subcaption}
\usepackage{wrapfig}
\usepackage{float}

\usepackage{amsthm}
\usepackage{amsmath}
\usepackage{thmtools,thm-restate}
\usepackage{mdframed}
\usepackage{mathtools}

\newtheorem{theorem}{Theorem}[section]
\newtheorem{lemma}[theorem]{Lemma}
\newtheorem{corollary}[theorem]{Corollary}
\newtheorem{proposition}[theorem]{Proposition}
\newtheorem{definition}[theorem]{Definition}

\newmdenv[
  topline=false,
  bottomline=false,
  skipabove=\topsep,
  skipbelow=\topsep
]{siderules}

\begin{document}
%%%%%%%%%%%%%%%%
% Outcomment only when entries are known. Otherwise leave as is and 
%   default values will be used.
%\setcounter{page}{1}
%\VOLUME{00}%
%\NO{0}%
%\MONTH{Xxxxx}% (month or a similar seasonal id)
%\YEAR{0000}% e.g., 2005
%\FIRSTPAGE{000}%
%\LASTPAGE{000}%
%\SHORTYEAR{00}% shortened year (two-digit)
%\ISSUE{0000} %
%\LONGFIRSTPAGE{0001} %
%\DOI{10.1287/xxxx.0000.0000}%

\title{Conditions for Stability in Strategic Matching} 
\author{James P. Bailey \\ Texas A\&M University \\ \EMAIL{jamespbailey@tamu.edu} \and Craig A. Tovey \\ Georgia Institute of Technology \\ \EMAIL{craig.tovey@isye.gatech.edu}}

\date{}

\maketitle

%\RUNAUTHOR{Bailey and Tovey}
%\RUNTITLE{}
%\TITLE{Conditions for Stability in Strategic Matching}
%\ARTICLEAUTHORS{%
%\AUTHOR{James P. Bailey}
%\AFF{Singapore University of Technology and Design, \EMAIL{james\_bailey@sutd.edu.sg}, \URL{}}
%\AUTHOR{Craig A. Tovey}
%\AFF{Georgia Institute of Technology, \EMAIL{cat@gatech.edu}, \URL{}}
%}

\begin{abstract}
We consider the stability of matchings when individuals strategically
submit preference information to a publicly known
algorithm.  Most pure Nash equilibria of the ensuing game yield a matching
that is unstable with respect to the individuals’ sincere preferences.  We
introduce a well-supported minimal dishonesty constraint, and obtain
conditions under which every pure Nash equilibrium
yields a matching that is stable with respect to the sincere
preferences.  The conditions on the matching algorithm are to be either
fully-randomized, or monotonic and independent of non-spouses (INS), an
IIA-like property. These conditions are significant because they support
the use of algorithms other than the Gale-Shapley (man-optimal) algorithm
for kidney exchange and other applications. We prove that the Gale-Shapley
algorithm always yields the woman-optimal matching when individuals are
minimally dishonest. However, we give a negative answer to one of Gusfield
and Irving’s open questions: there is no monotonic INS or
fully-randomized stable matching algorithm that is certain to yield the
egalitarian-optimal matching when individuals are strategic and minimally
dishonest. Finally, we show that these results extend to the student
placement problem, where women are polyandrous but must be honest but do
not extend to the admissions problem, where women are both polyandrous and
strategic.
\end{abstract}

% Fill in data. If unknown, outcomment the field
%\KEYWORDS{Strategic;  Stable Matching;  Stable Marriage;  College Admissions;  Student Placement;  Gale-Shapley, Minimal Dishonesty, Partial Honesty, Truth-Bias}
%\MSCCLASS{}
%\ORMSCLASS{Primary: ; secondary: }
%\HISTORY{}

%%%%%%%%%%% INTRO
\section{Introduction}\label{sec:Intro}
A matching between two disjoint sets of agents (men and women) is {\it
	stable} if for no man/woman pair does each prefer the other to the one
he/she is matched with. For any profile $\Pi$ of each man's preference
order over the women, and vice-versa, there exists at least one stable
matching.  Gale and Shapley \cite{GaleShapley62} proved this by inventing the first
stable matching algorithm, which on input $\Pi$ produces a matching that
is stable with respect to $\Pi$.  Their algorithm obtains the
``man-optimal'' stable matching, so called because it matches each man to
the woman he most prefers among all women to whom he is matched in at
least one stable matching.  There are several other polynomial time stable matching
algorithms \cite{Gusfield87,irving87};  the polytope (convex hull) of
stable matchings is polynomially separable \cite{VandeVate89,Rothblum92} and thus for many
objective criteria, an optimum stable matching can be found by linear
programming.  On the other hand, it is NP-hard to optimize with respect to
either the {\it sex-equal} and {\it median} criteria \cite{Kato93,Cheng2008}.

This paper considers the problem of obtaining a matching that is stable
with respect to the sincere preferences $\Pi$ -- henceforth called a {\it
	sincerely stable matching} -- when individuals strategically submit
preference data $\SPi$.  The standard way public choice mechanisms deal successfully with strategic behavior is to be strategy-proof, such that rational individuals will selfishly choose to be truthful.  The problem considered here is mathematically interesting because no stable matching algorithm guarantees rational individuals will behave honestly \cite{Irving86}.  Hence one must seek sincere stability without eliciting
honesty. 
To our knowledge there are no results in the literature where a public choice mechanism is proven to achieve its goal when individuals have incentive to be dishonest. 
How to obtain matchings that are stable with
respect to $\SPi$ poses a novel challenge because one has no direct access
to $\Pi$. One can only make inferences about $\Pi$ from $\SPi$.

We provide conditions under which individuals will, in equilibrium, submit
false data $\SPi$ that nonetheless yields a matching that is stable with
respect to $\Pi$.
Our conditions are that the individuals are 
minimally dishonest, and that the public matching algorithm be either fully
randomized (with respect to individual voters), or be monotonic and
independent of non-spouses (INS), an IIA-like property.  
We introduce minimal dishonesty as a Nash equilibrium refinement  and formally define it in 
Section
\ref{sec:Dishonest}. Moreover, we provide a wide array of studies that empirically support our Nash refinement in Section \ref{sec:averse}.
Informally, people are minimally dishonest if %(as supported by experimental observation}
they  lie only to the extent necessary to maximize their utility. 
A fully randomized algorithm is such that either an individual is assigned the same spouse in every stable matching, or there are at least two different potential spouses who will be assigned to the individual  with positive probability.
The monotonicity and INS properties are similar to rationality
criteria in voting theory.  All but a few of the stable matching algorithms
in the literature satisfy these properties.

%However, there are few results in literature where an
%algorithm is known to obtain good results when individuals have incentive
%to be dishonest.

Based on the interplay of optimality conditions and rationality
constraints, we characterize the ways in which an individual will
strategically submit preferences. Individually, these characterizations say
nothing about the stability of the outcome. However, we show that these
characterizations collectively imply that the output of every equilibrium
will be a sincerely stable matching even when individuals submit insincere
preferences.

Our main results, good, bad, and perverse,  are then as follows. (1) Under
the stated conditions, every pure Nash equilibrium produces a sincerely
stable matching.  This result is significant theoretically and because it
supports the use of algorithms other than the Gale-Shapley (man-optimal)
algorithm for applications of stable matching.  (2) There does not
exist a monotonic INS or fully-randomized stable matching algorithm that is
certain to yield the egalitarian-optimal matching when individuals are
strategic and minimally dishonest.  This answers a specific open question
of Gusfield and Irving’s in the negative. (3)  The Gale-Shapley
(man-optimal) algorithm always yields the woman-optimal matching when
individuals are locally minimally dishonest. (A weaker form of this
perverse result was proved by Gale and Sotomayer \cite{GaleSotomayor85}.)

Finally, we extend these results to the student placement
problem, where women are polyandrous but must be honest, but show they do not hold for  the admissions problem, where women are both polyandrous and strategic.
Roth has long insisted that the student placement and admissions problems
are qualitatively different.  Our contrasting results support Roth's claim \cite{Roth85},
and  improve insight into their differences.

\subsection{Related Literature on Social Aversion to Lying}\label{sec:averse}

We introduce and propose the notion of minimal dishonesty (formally defined in Section \ref{sec:Dishonest}) to constrain the ways in which an individual can strategically submit preference data.
This notion is supported by a substantial body of empirical evidence from the experimental economics and psychology literature that indicates people are averse to lying.  Gneezy \cite{Gneezy2005} experimentally finds that people do not lie unless there is a benefit.  Hurkens and Kartik \cite{WouldILie} perform additional experiments that confirm an aversion to lying, and show their and Gneezy's data to be consistent with the assumption that some people never lie and others always lie if it is to their benefit.  Charness and Dufwenberg \cite{charness2006promises} experimentally find an aversion to lying and show that it is consistent with guilt avoidance.  Battigalli \cite{battigalli2013deception} experimentally finds some contexts in which guilt is insufficient to explain aversion to deception.  Several papers report evidence of a ``pure'' i.e., context-independent, aversion to lying \cite{lopez2013people,charness2010bare} that is significant but not sufficient to fully explain experimental data.

The set of research results we have cited here is by no means exhaustive.  Two additional ones are of particular relevance to our concept of minimal dishonesty.    Mazar {\it et al.} \cite{mazar2008dishonesty} find that  ``people behave dishonestly enough to profit but
honestly enough to delude themselves of their own integrity.'' Lundquist {\it et al.} \cite{Lundquist2009aversion} find that people have an aversion to lying which increases with the size of the lie.  Both of these studies support our hypothesis that people will not lie more than is necessary to achieve a desirable outcome.  We further suggest that less dishonesty may require less  cognitive effort:  once one has determined a sufficient set of false information that will yield one's desired outcome, it is simpler to fill in the rest of the information truthfully than to invent additional lies.

Some experimental evidence is less confirmatory of our hypothesis.  Several studies, beginning with \cite{Gneezy2005},  have found an aversion to lying if doing so would disbenefit someone else substantially more than the benefit one would accrue.

More recently,  honesty conditions have been explored in the voting literature.  An individual is {\it partially honest} if she is completely honest unless dishonesty can yield a strictly preferred outcome \cite{Dutta10,Dutta12,Kartik14,Laslier17,Nunes15,Matsushima} (see also \cite{Obraztsova13}).
However, partial honesty, though mathematically implied by minimal dishonesty, evaluates honesty in a binary sense, which is inconsistent with the aforementioned experimental literature.  Moreover,  we prove in Section \ref{sec:partialhonest} that partial honesty is insufficient to assure sincere stability.

We may nonetheless seek some other condition that is weaker than minimal dishonesty yet still ensures sincere stability.  In Section \ref{sec:Dishonest} we define such a condition,  {\it locally minimal dishonesty}.  Informally, it forbids a preference order $p$ from being submitted if a simple swap  would yield a more honest $p'$ without decreasing utility.  Interestingly, this condition neither implies nor is implied by partial honesty.

%%%%%%%%%%%%%%%% DEFINITIONS
\section{Preliminaries.}\label{sec:definition}

An instance of the Stable Matching Problem entails finite nonempty sets of men $M$ and women $W$. 
Each $m\in M$ has strict preferences $\SP_m$ on  $W\cup\{m\}$. 
$\SP_m$ is a total ordering on $W\cup \{m\}$. 
Symmetrically, woman $w$ has strict preferences $\SP_w$ on the set $M\cup\{w\}$.    
The ordering $\Pi_m$ denotes $m$'s preferences on the set of women. 
The notation  $w_1\Pi_mw_2$ indicates that $m$ strictly prefers $w_1$ to $w_2$ and $m\Pi_m w$ indicates that $m$ strictly prefers $m$ to $w$ (he would rather be unmatched than matched to $w$). 
For readability, we express the full preference list $\SP_m$, $w_1 \SP_m w_2 \SP_m m \SP_m w_3 \SP_m w_4$, as ($\SP_m: \ w_1, w_2, \mathbf{m}, w_3, w_4$) and denote $\SP_{mk}$ as $m$'s $k$th favorite partner. 
The collection of all preferences $\Pi=\{\Pi_i\}_{i\in M\cup W}$ is the preference profile.   
Let ${\cal P}_i$ be the set of possible orderings for $i$ and ${\cal P}={\displaystyle\bigtimes_{i\in M\cup W}}{\cal P}_i$ be the set of all preference profiles.

A matching $\mu$ is a bijection from $M\cup W$ to itself such that $\mu(m)\in W\cup\{m\}$ and $\mu(w)\in M\cup\{w\}$ for each $m\in M$ and $w\in W$. 
Moreover, the relationship is symmetric; $\mu(i)=j$ if and only if $\mu(j)=i$. 
We denote $\mu(i)$ as the spouse of $i$.  
If $\mu(i)=i$ then $i$ is self-matched (equivalently unmatched). 
Individual $i$ strictly prefers $\mu_1$ to $\mu_2$ if and only if $\mu_1(i)\Pi_i\mu_2(i)$.
Denote the set of all matchings as ${\cal M}$.  

A matching $\mu$ is \emph{individually rational} with respect to $\Pi$  if $\mu(y)=y$ or $\mu(y){{\SP}_y}y$ for every individual $y$. 
If a matching $\mu$ is not individually rational then there is an individual $i$ who prefers being unmatched to the matching $\mu$ ($i\Pi_i\mu(i)$).  
As a result $\mu$ is considered unstable since $i$ would leave $\mu(i)$ to be single.

A pair $\{m,w\}\in M\times W$ is a blocking pair for $\mu$ with respect to ${\SP}$ if $w {{\SP}_m} \mu(m)$ and $m {{\SP}_w} \mu(w)$. 
If $\{m,w\}$ is a blocking pair for $\mu$, then $\mu$ is again unstable since $m$ and $w$ would leave their current spouses to be together. 
Stability is a necessary condition for any solution \cite{Knuth76,GusfieldIrving89,Manlove2013}.
The matching $\mu$ is \emph{stable} with respect to $\Pi$ if $\mu$ is individually rational and has no blocking pairs with respect to  $\Pi$. 
Equivalently, $\mu$ is stable if $y \SP_z \mu(z)$ implies $\mu(y) \SP_y z$. 

A deterministic stable matching mechanism is a function $r: {\cal P}\to {\cal M}$ where $r(\Pi)$ is a stable matching for all $\Pi\in {\cal P}$. 
In this paper we also consider randomized stable matching mechanisms.  In general, a stable matching mechanism is a function $r: {\cal P}\to [0,1]^{|{\cal M}|}$ where $r_t(\Pi)$ is the probability of selecting the matching $\mu_t\in {\cal M}$. If $\mu_t$ is unstable with respect to $\Pi$, then $r_t(\Pi)=0$. 
For a given r, let $p_{ij}(\Pi)$ denote the probability that $r(\Pi)$ will match $i$ to $j$.  
Formally, $p_{ij}(\Pi)=\sum_{\mu_t: \mu_t(i)=j}r_t(\Pi)$. 

The ability to select a stable matchings relies on the assumption that individuals are truthfully reporting their preferences. 
However, no stable matching mechanism guarantees strategy-proofness \cite{Roth1984}. 
Individuals might submit a strategic profile $\SPi$ that is not equal to the sincere profile $\Pi$. 
To understand the outcome of this strategic behavior, examine a normal-form game with complete information where individuals can submit whichever ordering they like despite the fact that players have common knowledge about the sincere profile $\Pi$.  Denote this game as the Strategic Stable Matching Game (SSM).

\hrulefill

\centerline{\bf Strategic Stable Matching Game (SSM)}
\begin{itemize} \item Each individual has complete information of the sincere preference profile $\Pi=\{\Pi_i\}_{i\in M\cup W}$.
	\item To play the game, individual $i$ submits putative preference data $\spi_i\in {\cal P}_i$.  The collection of all submitted data is denoted $\SPi$.
	\item It is common knowledge that a central decision mechanism will select the outcome $r(\SPi)$ which is stable with respect to $\SPi$.
	\item Individual $i$ evaluates $r(\SPi)$ according to $i$'s partner(s) in the matching(s) $r(\SPi)$ and $i$'s sincere preferences $\Pi_i$. \end{itemize}
\hrulefill

\begin{definition}A matching $\mu$ is a sincerely  (respectively putatively) stable matching is if it is stable with respect to $\Pi$ (respectively $\SPi$). \end{definition}

There are many different stable matching mechanisms. In this paper we focus on mechanisms that are monotonic and independent of non-spouses, or Fully-Randomized.  Most commonly used algorithms satisfy either the first two properties or the third.

\begin{definition}[Monotonicity]For given $\SP$, $i$ and $j$ let $\SP'$ be the profile obtained by moving $j$ up one position in the ordering $\SP_i$.  A stable matching mechanism $r(\cdot)$ is monotonic iff $p_{ij}(\Pi')\geq p_{ij}(\Pi)$ for all $\SP$, $i$ and $j$.  
\end{definition}

\begin{definition}[Independence of Non-Spouses (INS) ]
	For given $\SP$, $i$, $j$, and $k$ let $\SP'$ be the profile obtained by moving $j$ up one position in the ordering $\SP_i$ thereby moving $k$ down one position. 
	A stable matching mechanism $r(\cdot)$ is independent of non-spouses (INS) iff $p_{ik}(\SP)=0$ implies $p_{il}(\SP')\leq p_{il}(\SP)$ for $l\neq j$ for all $\SP$, $i$, and $j$. 
\end{definition}

Monotonicity guarantees that moving if individual $i$ moves  individual $j$ up in their preference list cannot decrease their probability of being matched to individual $j$. 
Independence of non-spouses (INS) guarantees that only $i$'s probability of matching $j$ can increase if $j$ replaces a non-spouse -- someone $i$ will not be matched with ($P_{ik}(\SP)=0$). 
The rationale behind INS is only $i$'s probability of matching $j$ should increase since $i$ only improved relative position of individual $j$. 
We remark that the condition $p(ik)=0$ is necessary for this property. 
If $p(ik)>0$ then there is a stable matching between $i$ and $k$. 
However, moving $j$ one position before $k$ may cause $\{i,j\}$ to become a blocking pair for this matching causing $i$'s probability of matching with $k$ to decrease.
As such, to ensure that $i$'s probability of being matched is one, $p_{il}$ may need to increase for some $l\neq j$.  
In addition to being natural, both properties occur in most standard matching algorithms.

\begin{definition}[Fully-Randomized] Mechanism $r(\cdot)$ is \emph{Fully-Randomized}  iff for all $\SP$ and each pair of  individuals $i$ and $j$ either (1) $p_{ij}(\Pi)<1$ or (2) $\mu(i)=j$ for all $\mu$ that are stable with respect to $\Pi$. \end{definition}

The most commonly used algorithm for stable matching is the Gale-Shapley (G-S) algorithm.
It is easily shown to be both monotonic and INS.  
In addition to proving results specifically for G-S, we use it as a tool in several of our results.
The G-S algorithm is easily described in words:
Each man and woman begins self-matched ($\mu(i)\gets i \ \forall i$).  
If there is a self-matched man $m$ that hasn't proposed to every woman he is willing to match, then he proposes to his most preferred woman $w$ that has not rejected him. 
If $w$ prefers her current match $\mu(w)$ to $m$ (i.e., if $\mu(w) \Pi_w m$), then she immediately rejects him.
Otherwise, $w$ prefers $m$ and she rejects her current match causing $\mu(w)$ to be self-matched ($\mu(\mu(w))\gets \mu(w)$), and she matches with $m$ ($\mu(w) \gets m$ and $\mu(m) \gets w$). 
The algorithm terminates once each man is either matched to someone in ${W}$ or has been rejected by every woman he is willing to match.

\subsection{Equilibria of SSM.}\label{sec:Equil}

Given that it is well known that there is no strategy-proof stable matching mechanism, it is unsurprising that there are equilibria with sincerely unstable outcomes. 
Given a deterministic stable matching mechanism, Alcalde and S\"{o}nmez  have shown that the set of equilibria correspond to the set of individually rational matchings \cite{Alcalde96,Sonmez97}. 
We generalize this result to randomized algorithms. 

\begin{lemma} \label{lem:deterministic} Let $r(\cdot)$ be a stable matching mechanism. If $\SPi$ is a pure strategy Nash equilibrium for $\Pi\in {\cal P}$ then the outcome $r(\SPi)$ is selected deterministically. \end{lemma}

\begin{proof}
For contradiction, suppose there is an individual $i$ where $p_{ij}(\SPi)<1$ for all $j$. 
Let $k$ be $i$'s most preferred partner where $p_{ik}(\SPi)>0$ and let $\mu$ be a matching in the support of $r(\SPi)$ where $\mu(i)=k$. 
If individual $i$ instead submits $\spi_i'$ where $k\spi_i' i$ and $i\spi_i l$ for all $l\notin \{i,k\}$ then $\mu$ remains stable.  
Moreover, any matching $\mu'$ where $\mu'(i)\notin\{k,i\}$ is not individually rational and therefore not stable. 
By the Rural Hospital theorem \cite{GaleSotomayor85,Roth86}, individual $i$ is matched to $k$ in every matching stable with respect to $[\SPi_{-i},\spi'_i]$ where $[\SPi_{-i},\spi'_i]$ is the profile obtained by replacing $\spi_i$ with $\spi'_i$ in the profile $\SPi$.  
Therefore $i$ can update $\spi_i$ to obtain a better outcome contradicting that $\SPi$ is a Nash equilibrium. \end{proof}  

Now that the outcome has been shown to be selected deterministically, individual rationality follows identically to \cite{Alcalde96,Sonmez97}. 

\begin{corollary}\label{thm:AllIndividuallyRational} Let $r(\cdot)$ be a stable matching mechanism and $\Pi\in {\cal P}$ be arbitrary. There exists a pure strategy Nash equilibrium whose outcome is $\mu$ if and only if $\mu$ is individually rational with respect to $\Pi$. 
\end{corollary}

%%%%%%%%%%%%%% REFINEMENT

\subsection{The Minimal Dishonesty Refinement.}\label{sec:Dishonest}

The proof of Corollary \ref{thm:AllIndividuallyRational} can require  individuals to lie in absurd ways to obtain results that they least prefer.  
For instance, Corollary \ref{thm:AllIndividuallyRational} permits everyone reporting that they would be prefer to be unmatched to be a perfectly reasonable outcome regardless of the sincere profile. 
This is unsatisfactory in a predictive sense - any individually rational matching can eventuate - and in a normative sense - as discussed in Section \ref{sec:averse} individuals tend to lie only if they benefit.

To improve the theory, we introduce and propose a refinement to the set of equilibria to those where individuals are \emph{minimally dishonest}.  
Empirical and logical  justifications for this refinement were given in Section \ref{sec:averse}.  
Minimal dishonesty assumes that individuals prefer to be as little dishonest as possible without worsening their outcome.
To formalize and apply this concept, one needs a way to measure dishonesty. 
We will use the Bubble Sort or Kendall Tau distance -- the most common way to evaluate the distance between two ordered lists.

\begin{definition}[Bubble Sort or Kendall Tau Distance]\label{def:metric}
	Let $\Pi^1_m$ and $\Pi^2_m$  be two preference lists over a set.  Then the {\it Kendall Tau distance} between $\Pi^1_m$ and $\Pi^2_m$ is
	\begin{align}\label{eqn:metric}
	K(\Pi^1_m,\Pi^2_m) \equiv \left|\left\{ \{i,j\}: i\neq j; i \Pi^1_m j \text{ but } j \Pi^2_m i \right\}\right|.
	\end{align}
\end{definition}

The formal definition of minimal dishonesty can now be stated succinctly in the notation of Definition
\ref{def:metric}.

\begin{definition}[Minimally Dishonest]\label{def:GlobalMD}  Let $\Pi$  be the sincere preferences and let $\SPi$ be an equilibrium in the Strategic Stable Matching game where $r(\SPi)=\mu_1$.   Individual  $y$ is \emph{minimally dishonest} if $K(\PP'_y,\SP_y)< K(\PP_y,\SP_y)$ implies $\mu_1(y) \SP_y \mu_2(y)$ for some $\mu_2(y)\in supp\left(r\left([\SPi_{-y},\spi'_y]\right)\right)$.
\end{definition}

If there is a ${\PP}'_y$ such that $K({\PP}'_y,\SP_y)<K({\PP}_y,\SP_y)$ and $y$ does not sincerely prefer $r({\SPi})$ to  $r([\SPi_{-y},\spi'_y])$, then $y$ can obtain at least as good a result by submitting the more honest ${\PP}'_y$.  
Equivalently, submitting the more honest ${\PP}'_y$ results in positive probability of obtaining the less preferred outcome $\mu_2(y)\in supp\left(r\left([\SPi_{-y},\spi'_y]\right)\right)$.
If an individual is able to be more honest and obtain at least as good a result, we assume the individual would do so since individuals prefer being honest. Thus we only examine \emph{minimally dishonest equilibria} -- equilibria where every individual is minimally dishonest.

We also show all of our results for a weaker version of minimal dishonesty.  To be \emph{locally minimally dishonest} there may not be a pairwise change to the putative preference list that is more honest and results in as good an outcome.

\begin{definition}[Locally Minimally Dishonest]\label{def:LocalMD}  Let $\Pi$  be the sincere preferences and let $\SPi$ be an equilibrium in the Strategic Stable Matching game where $r(\SPi)=\mu_1$.   Individual  $y$ is \emph{minimally dishonest} if $K(\PP'_y,\SP_y)< K(\PP_y,\SP_y)$ implies $\mu_1(y) \SP_y \mu_2(y)\neq \mu_1(y)$ for some $\mu_2(y)\in supp\left(r\left([\SPi_{-y},\spi'_y]\right)\right)$ for all $\PP'_y$ obtained by swapping two individuals in the list $\PP_y$.
\end{definition}

Once again, if the condition fails to hold, then $m$ could obtain at least as good  a result when the more honest $\PP'_y$ is submitted.  Woman $w$ is locally minimally dishonest if symmetric conditions hold for her.

If an individual is not locally minimally dishonest then the individual is also not minimally dishonest.  However, an individual may be locally minimally dishonest without being minimally dishonest. 
Thus every minimally dishonest equilibrium is also a locally minimally dishonest equilibrium.
Therefore, if a result describes a property of all (locally) minimally dishonest equilibria it suffices to show this when individuals are locally minimally dishonest. 
On the other hand, if we show there is an equilibrium with a certain property, then it suffices to show this when individuals are minimally dishonest. 
Therefore all of our results will hold for both minimal dishonesty and locally minimally dishonesty.

\subsection{Other Refinements on the Nash Equilibria.}\label{sec:other}

In this paper we also discuss two other refinements to the set of Nash equilibria. 
The first is known as a \emph{truncation refinement} where individuals act strategically but only by truncating their sincere preferences, i.e., for all $w,w'\in W$, $w \PP_m w'$ implies $w \SP_m w'$, however the position of $\mathbf{m}$ may change.
This refinement is known to yield sincerely stable matchings \cite{Roth99, Ehlers08}. 
However, this refinement is not well supported by experimental evidence.
Moreover, we show the (locally) minimally dishonest refinement and the truncation refinement yield distinct equilibria. 

%The second refinement we discuss is from the voting literature and known as \emph{partial honesty} where an individual is completely honest unless it negatively impacts her valuation of the outcome.
%Unlike minimal dishonesty, it evaluates honesty in a binary sense which is inconsistent with the literature in Section \ref{sec:averse}. 
%Moreover, we show that this refinement does not guarantee stable matchings. 
%There is, however, a recent variant of truth-bias introduced in \cite{Obraztsova17} where individual utilities are penalized by the size of the lie they tell.
%In the setting of stable matching, it is relatively straightforward to show that this definition is equivalent to minimal dishonesty.
%However, when individuals select strategies from a continuous space, this variant of truth-bias may result in irrational behavior unlike minimal dishonesty \cite{BaileyCOMSOC18}.

The second refinement we discuss is from the voting literature. 
An individual is \emph{partial honesty} if she is completely honest unless she can obtain a strictly better outcome by lying.
Locally minimal dishonesty and partial honesty can be viewed as opposite ends of a spectrum;
locally minimal dishonesty only considers preference lists that are slightly more honest while partial honesty considers only the most honest preference list, i.e., the sincere preferences. 
Minimal dishonesty considers all preference lists and therefore implies both locally minimal dishonesty and partial honesty. 
Interestingly however, partial honesty and locally minimal dishonesty together do not imply minimal dishonesty.

%Unlike minimal dishonesty, partial honesty evaluates truthfulness in a binary sense which is inconsistent with the literature in Section \ref{sec:averse}. 
%Moreover, we show that this refinement does not guarantee stable matchings. 
There is a recent variant of partial honesty introduced in \cite{Obraztsova17} where individual utilities are penalized by the size of the lie they tell.
In the setting of stable matching, it is relatively straightforward to show that this definition is equivalent to minimal dishonesty.
However, when individuals select strategies from a continuous space, this variant of partial honesty may result in irrational behavior unlike minimal dishonesty \cite{BaileyCOMSOC18}.

\section{Monotonic INS stable matching mechanisms.}
\label{sec:Mono}

Our first contribution is showing that monotonicity and INS are sufficient to gaurantee stability when individuals are minimally or locally minimally dishonest.

\begin{restatable}{theorem}{Main}\label{thm:Main} Suppose individuals are strategic and (locally) minimally dishonest. If $r(\cdot)$ is monotonic and INS, then every outcome of SSM is sincerely stable.  
\end{restatable}

We first  establish necessary conditions on the strategic $\PP$ as they relate to the sincere $\SP$.
Since every minimally dishonest equilibrium is also a locally minimally dishonest equilibrium, it suffices to show each result when individuals are locally minimally dishonest.

\begin{lemma} Let $r(\cdot)$ be as in the statement of Theorem \ref{thm:Main}. Given the sincere $\SP$ and (locally) minimally dishonest Nash equilibrium $\PP$, let $\mu=r(\PP)$. Let $y_1$ and $y_2\neq \mu(z)$ be adjacent in $\PP_z$.  If $y_1 \SP_z \mu(z)$ and $y_1 \SP_z y_2$, then $y_1 \PP y_2$. \label{lem:PartialOrder}
\end{lemma}

\begin{proof}[Proof of Lemma \ref{lem:PartialOrder}] For contradiction suppose $y_2\PP_z y_1$.  
Let $\PP'$ be the profile obtained after $z$ switches $y_1$ and $y_2$ in the ordering $\PP_z$. 
By adjacency of $y_1$ and $y_2$, $K(\PP',\SP)=K(\PP,\SP)-1$.
Minimal dishonesty implies that $z$ should obtain a strictly worse outcome with this update.  
However, this is not the case:
By Lemma \ref{lem:deterministic}, $p_{z\mu(z)}(\PP)=1$. 
If $y_1=\mu(z)$, monotonicity implies $p_{zy_1}(\PP)\geq p_{zy_1}(\PP')=1$.
If $y_1\neq \mu(z)$, INS implies $p_{zy_3}(\PP')\leq p_{zy}(\PP)=0$ for all $y\notin \{y_1, y_2, \mu(z)\}$ and monotonicity implies $0=p_{zy_2}(\PP) \geq p_{zy_2}(\PP')$. 
In both cases, $z$ is matched to either $y_1$ or $\mu(z)$ in the outcome $r(\PP')$, implying that $z$ obtains at least as good a spouse.  
This contradicts locally minimal dishonesty.
\end{proof}  

Lemma \ref{lem:PartialOrder} indicates that $z$ is relatively honest about the the individuals $z$ prefers to $\mu(z)$. 
Formally, if $\mu(z)$ is $z$'s kth most preferred partner, then Lemma \ref{lem:PartialOrder} only guarantees that the first $k$ elements of $\PP_z$ are a permutation of the first $k$ elements of $\SP_z$.

%\begin{corollary}
%	Let $r(\cdot)$, $\SP$ and $\PP$ be as in the statement of Lemma \ref{lem:PartialOrder}. If $y_1 \SP_z \mu(z)$, then $y_1 \PP_z \mu(z)$. \label{cor:partOrder}
%\end{corollary}
%\begin{proof}
%	For contradiction, suppose the statement of the corollary is not true and let $y_1$ be the highest ranked individual in $\PP_z$ where $y_1 \SP_z \mu(z)$ and $\mu(z) \PP_z y_1$.
%	Let $y_2$ be the individual appearing directly prior to $y_1$ in $\PP_z$ (possibly $y_2=\mu(z)$). 
%	By selection of $y_1$, $y_1 \SP_z y_2$, and by Lemma \ref{lem:PartialOrder}, $y_1 \PP_z y_2$, a contradiction since $y_2$ is one position prior to $y_1$ in $\PP_z$.  
%\end{proof}
%

\begin{corollary}
	Let $r(\cdot)$, $\SP$ and $\PP$ be as in the statement of Lemma \ref{lem:PartialOrder}.  Only one matching is stable with respect to $\PP$.  \label{cor:one}
\end{corollary}

\begin{proof} Suppose there is a $\mu_2\neq \mu_1=r(\PP)$ that is stable with respect to $\PP$.  Since $\mu_2$ is stable, there exists a $z$ such that $\mu_2(z) \PP_z \mu_1(z)\neq \mu_2(z)$ \cite{Knuth76}, i.e., with respect to the putative $\PP$, there is at least one individual $z$ that reports $\mu_2$ is preferred to $\mu_1$.  
By Lemma \ref{lem:PartialOrder}, $\mu_2(z) \SP_z \mu_1(z)\neq \mu_2(z)$.  
However, this implies that $z$ can obtain a strictly better outcome by submitting $\mu_2(z) \PP'_z z \PP'_z y$ for all $y\neq \mu_2(z)$ contradicting that $\PP$ is a Nash equilibrium. 
\end{proof}

Using this corollary, we refine Lemma \ref{lem:PartialOrder} to show $\SP_z$ and $\PP_z$ will be identical up to $z$'s assigned match $\mu(x)$, i.e., $z$ honestly reveals their preferences until their match $\mu(z)$. 

\begin{corollary}
	Let $r(\cdot)$, $\SP$ and $\PP$ be as in the statement of Lemma \ref{lem:PartialOrder} and let $\mu=r(\PP)$. If $y_1 \SP_z \mu(z)$ and $y_1 \SP_z y_2$, then $y_1 \PP y_2$. \label{cor:CompleteOrder}
\end{corollary}

\begin{proof} For contradiction, suppose $y_1 \SP_z \mu(z)$ and $y_1 \SP_z y_2$ but $y_2 \PP_z y_1$. 
It suffices to consider only $y_1$ and $y_2$ adjacent in $\PP_z$.  
If not, there is a $y_3\notin \{y_1,y_2\}$ such that $y_2 \PP_z y_3 \PP_z y_1$. 
Either $y_3\SP_z y_1$ and we can proceed with $y_2$ and $y_3$ or $y_1 \SP_z y_3$ and we can proceed with $y_1$ and $y_3$. Thus, thus we consider $y_1$ adjacent to $y_2$ in $\PP_z$.

Lemma \ref{lem:PartialOrder} implies the statement of the corollary if $y_2\neq \mu(z)$ and therefore we only consider $y_2=\mu(z)$.    
Let $\PP'$ be the profile obtained when $z$ swaps $y_1$ and $\mu(z)$ in $\PP_z$.  
If $r(\PP')$ assigns $z$ to either $y_1$ or $\mu(z)$ then we have a contradiction to locally minimal dishonesty. 
Therefore it suffices to show $r(\PP')$ assigns $z$ to either $y_1$ or $\mu(z)$.  

Let $\mu'$ be such that $\mu'(z)\notin \{y_1,\mu(z)\}$.  
By Corollary \ref{cor:one}, $\mu'$ is unstable with respect to $\PP$. 
Therefore there exists $a$ and $b\neq \mu'(a)$ such that $b \PP_a \mu'(a)$ and $a \PP_b \mu'(b)$.  
If $z\neq a$ then $b \PP_a \mu'(a)$ implies $b \PP'_a \mu'(a)$ since $\PP_a=\PP'_a$.  
If $z=a$ then $b \PP_z \mu'(z)$ implies $b \PP'_z \mu'(z)$ since $\mu'(z)\notin \{y_1,\mu(z)\}$ and since $\PP'_z$ is obtained from $\PP_z$ by swapping the adjacent pair $\{y_1,\mu(z)\}$.  
In both cases, $\mu'$ is also unstable with respect to $\PP'$.  
Thus if $\mu'$ is stable with respect to $\PP'$ then $\mu'(z)\in \{y_1,\mu(z)\}$.
As result, swapping $y_1$ and $\mu(z)$ causes $z$ to be matched with either $y_1$ or $\mu(z)$, completing the proof of the corollary. 
\end{proof}

Thus at every minimally and locally minimally dishonest equilibrium, every individual is honest up to the spouse assigned to them by $r(\SPi)$.  
We now proceed to prove Theorem \ref{thm:Main}.

\begin{proof}[Proof of Theorem \ref{thm:Main}.] Let $\PP$ be a minimally dishonest Nash equilibrium for the sincere $\SP$. 
Let $r(\PP)=\mu$ be the matching assigned. We now show that $\mu$ is stable with respect to $\SP$. 
By Corollary \ref{thm:AllIndividuallyRational}, $\mu$ is individually rational with respect to $\SP$. 

Next, for contradiction, suppose that $\mu$ is not stable with respect to $\Pi$ and there is a pair $\{m,w\}$ where $w \SP_m \mu(m)$ and $m \SP_w \mu(w)$.  
Taking $w=y_1$, $m=z$, and $y_2=\mu(m)$, Corollary \ref{cor:CompleteOrder} implies $w \PP_m \mu(m)$. 
Similarly $m\PP_w\mu(w)$ and therefore $\{m,w\}$ blocks $\mu$ with respect to $\PP$ contradicting that $\mu=r(\PP)$. 
\end{proof}

Theorem \ref{thm:Main} shows that every monotonic INS stable matching algorithm achieves stability when individuals are strategic and (locally) minimally dishonest. 
The most desirable monotonic INS stable matching is the egalitarian stable matching -- a procedure that places equal weight on the preferences of men and women. 
Gusfield and Irving \cite{GusfieldIrving89} specifically call for an algorithm that guarantees an egalitarian stable matching. 
We provide a negative result, showing there is no monotonic INS stable matching algorithm that always outputs an egalitarian stable matching when individuals are strategic, even if they are (locally) minimally dishonest.

\begin{restatable}{theorem}{negativeresult}
	There does not exist a monotonic INS $r(\cdot)$ where (i) there is always a (locally) minimally dishonest equilibrium, and (ii) every (locally) minimally dishonest equilibrium yields a sincere egalitarian stable matching.\label{thm:Neg}
\end{restatable}

\begin{proof} We begin by describing a set of preferences $\SP^1$ with a single egalitarian matching $\mu_2$ and consider a monotonic INS $r(\cdot)$ such that there exists a (locally) minimally dishonest equilibrium $\PP^1$ where $r(\PP^1)=\mu_2$.  
Using (locally) minimal dishonesty, monotonicity, and INS we are able to determine $\PP^1$.  
We then create $\SP^2$ by modifying $\SP^1$ slightly so that $\mu_2$ is stable but not egalitarian with respect to $\SP^2$. 
We also create $\PP^2$ by modifying $\PP^1$ slightly so that $\PP^2$ is a (locally) minimally dishonest equilibrium for $\SP^2$ where $r(\PP^2)=\mu_2$. 
Thus there is a (locally) minimally dishonest equilibrium that does not yield a sincere egalitarian stable matching thereby completing the proof of the theorem.

\begin{table}[ht]	
	\centering\caption	
	{Sincere Preferences $\SP^1$ for Theorem \ref{thm:Neg}. \label{tab:1Neg}}	
	{\begin{tabular}{|c|c|}	
			\hline
			$\SP^1_{m_1}: \ w_1,w_2,w_4,w_3,\mathbf{m_1}$&
			$\SP^1_{w_1}: \ m_2,m_3,m_4,m_1,\mathbf{w_1}$\\
			$\SP^1_{m_2}: \ w_2,w_3,w_4,w_1,\mathbf{m_2}$&
			$\SP^1_{w_2}: \ m_3,m_1,m_4,m_2,\mathbf{w_2}$\\
			$\SP^1_{m_3}: \ w_3,w_1,w_4,w_2,\mathbf{m_3}$&
			$\SP^1_{w_3}: \ m_1,m_2,m_4,m_3,\mathbf{w_3}$\\
			$\SP^1_{m_4}: \ w_4,w_1,w_2,w_3,\mathbf{m_4}$&
			$\SP^1_{w_4}: \ m_4,m_1,m_2,m_3,\mathbf{w_4}$\\
			\hline
	\end{tabular}}	
	{}	
\end{table}	

\begin{table}[ht]	
	\centering\caption	
	{Stable Matchings with respect to $\SP^1$ and $\SP^2$. \label{tab:NegStable}}	
	{\begin{tabular}{|r | l |}	
			\hline
			$\mu_1$&$  \mu_1(m_1)=w_1, \ \mu_1(m_2)=w_2, \ \mu_1(m_3)=w_3, \ \mu_1(m_4)=w_4$\\
			$\mu_2$&$  \mu_2(m_1)=w_2, \ \mu_2(m_2)=w_3, \ \mu_2(m_3)=w_1, \ \mu_2(m_4)=w_4$\\
			$\mu_3$&$  \mu_3(m_1)=w_3, \ \mu_3(m_2)=w_1, \ \mu_3(m_3)=w_2, \ \mu_3(m_4)=w_4$\\
			\hline
	\end{tabular}}	
	{}	
\end{table}	

Consider the sincere preferences in Table \ref{tab:1Neg}.
The set of stable matchings is given in Table \ref{tab:NegStable}.
With respect to $\SP^1$, $\mu_2$ is the only egalitarian stable matching.  Let $r(\cdot)$ be a monotonic INS stable matching mechanism where there is a (locally) minimally dishonest equilibrium $\PP^1$ such that $r(\PP^1)=\mu_2$. 

\indent By Corollary \ref{cor:CompleteOrder}, $\PP_y^1$ and $\SP^1_y$ agree on the first two elements for $y\in \{m_1,m_2,m_3,w_1,w_2,w_3\}$ and $\PP_y$ and $\SP^1_y$ agree on the first element for $y\in \{m_4,w_4\}$. 
Moreover since $m_4$ and $w_4$ are each other's first choice, $m_4$ will be assigned $w_4$ regardless of the remainder of $\PP^1_{m_4}$.  
Therefore by (locally) minimal dishonesty, $\PP_{m_4}^1=\SP_{m_4}^1$.  
Symmetrically $\PP_{w_4}^1=\SP_{w_4}^1$. Therefore $\PP^1$ satisfies the relationship shown in Table \ref{tab:P1Start}. 

\begin{table}[ht]	
	\centering\caption	
	{Beginning of Putative Preferences $\PP^1$. \label{tab:P1Start}}
	{\begin{tabular}{|r l |r l |}	
			\hline
			$\PP^1_{m_1}:$&$ w_1,w_2,\mathbf{m_1}$&
			$\PP^1_{w_1}:$&$ m_2,m_3,\mathbf{w_1}$\\
			$\PP^1_{m_2}:$&$ w_2,w_3,\mathbf{m_2}$&
			$\PP^1_{w_2}:$&$ m_3,m_1,\mathbf{w_2}$\\
			$\PP^1_{m_3}:$&$ w_3,w_1,\mathbf{m_3}$&
			$\PP^1_{w_3}:$&$ m_1,m_2,\mathbf{w_3}$\\
			$\PP^1_{m_4}:$&$ w_4,w_1,w_2,w_3,\mathbf{m_4}$&
			$\PP^1_{w_4}:$&$ m_4,m_1,m_2,m_3,\mathbf{w_4}$\\
			\hline
	\end{tabular}}
	{}
\end{table}	

Next, at least one woman excludes her least preferred man.  
Otherwise the matching $\mu_1$ is stable with respect to $\PP$, a contradiction to Corollary \ref{cor:one}.  
Without loss of generality we assume $w_1$ indicates she is unwilling to match $m_1$.  
We also assume $m_1$ indicates he is unwilling to match $w_3$ and the preferences satisfy the relationship in Table \ref{tab:P1End}.

\begin{table}[ht]	
	\centering\caption	
	{Updated Putative Preferences $\PP^1$. \label{tab:P1End}}
	{\begin{tabular}{|r l |r l |}	
			\hline
			$\PP^1_{m_1}:$&$w_1,w_2,\mathbf{m_1},w_3$&
			$\PP^1_{w_1}:$&$m_2,m_3,\mathbf{w_1},m_1$\\
			$\PP^1_{m_2}:$&$w_2,w_3,\mathbf{m_2}$&
			$\PP^1_{w_2}:$&$m_3,m_1,\mathbf{w_2}$\\
			$\PP^1_{m_3}:$&$w_3,w_1,\mathbf{m_3}$&
			$\PP^1_{w_3}:$&$m_1,m_2,\mathbf{w_3}$\\
			$\PP^1_{m_4}:$&$w_4,w_1,w_2,w_3,\mathbf{m_4}$&
			$\PP^1_{w_4}:$&$m_4,m_1,m_2,m_3,\mathbf{w_4}$\\
			\hline
	\end{tabular}}
	{}	
\end{table}	

Finally we claim that $m_2, m_3, w_2,$ and $w_3$ are honest. Let $\PP$ be any profile that matches $\PP^1$ in the first two entries for each individual and agrees with Table \ref{tab:P1End}. Next, we run the Gale-Shapley algorithm (Man-Optimal algorithm) on $\PP$:
\begin{align}
&m_1 \text{ proposes to } w_1: && w_1 \text{ declines }	 	&&(w_1 \PP_{w_1} m_1)\\
&m_1 \text{ proposes to } w_2: && w_2 \text{ accepts } 		&&(m_1 \PP_{w_2} w_2)\\
&m_2 \text{ proposes to } w_2: && w_2 \text{ declines} 		&&(m_1 \PP_{w_2} m_2)\\
&m_2 \text{ proposes to } w_3: && w_3 \text{ accepts} 		&&(m_2 \PP_{w_3} w_3)\\
&m_3 \text{ proposes to } w_3: && w_3 \text{ declines}	 	&&(m_2 \PP_{w_3} m_3)\\
&m_3 \text{ proposes to } w_1: && w_1 \text{ accepts}	 	&&(m_3 \PP_{w_1} w_1)\\
&m_4 \text{ proposes to } w_4: && w_4 \text{ accepts}		&&(m_4 \PP_{w_4} w_4)
\end{align} 
Therefore we can conclude that $\mu_2$ is the man-optimal matching for $\PP$ without knowing the remaining details of $\PP$.  
Similarly, $\mu_2$ is woman-optimal and and we conclude $\mu_2$ is the only matching stable with respect to $\PP$ regardless of how the remainder of how $m_2, m_3, w_2,$ and $w_3$ fill out their remaining preferences.  
Thus, by (locally) minimal dishonesty, each of these individual are honest completing the proof of the claim.

The preferences are now given by Table  \ref{tab:P1}.  
Man $m_1$ and woman $w_1$'s are one of the lists given in Table \ref{tab:P1alt}. 
Regardless of which list $m_1$ submits, $m_1$ can be more honest by swapping $\mathbf{m_1}$ and $w_4$ and submitting $\SP_{m_1}$.   
Only $\mu_2$ and $\mu_3$ are stable with respect to $[\PP_{-m_1}^1, \SP_{m_1}^1]$ and therefore by (locally) minimal dishonesty $\mu_3 \in supp\left(r\left( [\PP_{-m_1}^1, \SP_{m_1}^1] \right)\right)$. 
Symmetrically, $\mu_1 \in supp\left(r\left( [\PP_{-w_1}^1, \SP_{w_1}^1] \right)\right)$.  

\begin{table}[ht]	
	\centering\caption	
	{Putative Preferences $\PP^1$. \label{tab:P1}}	
	{\begin{tabular}{|r l |r l |}	
			\hline
			$\PP^1_{m_1}:$&$w_1,w_2,\mathbf{m_1},w_3$&
			$\PP^1_{w_1}:$&$m_2,m_3,\mathbf{w_1},m_1$\\
			$\PP^1_{m_2}:$&$w_2,w_3,w_4,w_1,\mathbf{m_2}$&
			$\PP^1_{w_2}:$&$m_3,m_1,m_4,m_2,\mathbf{w_2}$\\
			$\PP^1_{m_3}:$&$w_3,w_1,w_4,w_2,\mathbf{m_3}$&
			$\PP^1_{w_3}:$&$m_1,m_2,m_4,m_3,\mathbf{w_3}$\\
			$\PP^1_{m_4}:$&$w_4,w_1,w_2,w_3,\mathbf{m_4}$&
			$\PP^1_{w_4}:$&$m_4,m_1,m_2,m_3,\mathbf{w_4}$\\
			\hline
	\end{tabular} }	
	{}	
\end{table}	

\begin{table}[ht]	
	\centering\caption	
	{Putative Preferences $\PP^1$ for $m_1$ and $w_1$. \label{tab:P1alt}}	
	{\begin{tabular}{|r l |r l |}	
			\hline
			$\PP^1_{m_1}:$&$w_1,w_2,w_4,\mathbf{m_1},w_3$&
			$\PP^1_{w_1}:$&$m_2,m_3,m_4,\mathbf{w_1},m_1$\\
			$\PP^1_{m_1}:$&$w_1,w_2,\mathbf{m_1},w_3,w_4$&
			$\PP^1_{w_1}:$&$m_2,m_3,\mathbf{w_1},m_1,m_4$\\
			\hline
	\end{tabular} }	
	{}	
\end{table}	

This condition on $r(\cdot)$ is sufficient to guarantee that the outcome of SSM using $r(\cdot)$ is not always a sincerely stable egalitarian matching.  
Specifically, we present the sincere preferences, $\Pi_2$, in Table \ref{tab:SP2} and a (locally) minimally dishonest equilibrium that selects a non-egalitarian stable matching.

\begin{table}[h]	
	\centering\caption	
	{Sincere Preferences $\SP^2$ for Theorem \ref{thm:Neg}.\label{tab:SP2}}	
	{\begin{tabular}{|r l |r l |}	
			\hline
			$\SP^2_{m_1}:w_1,w_4,w_2,w_3,\mathbf{m_1}$&
			$\SP^2_{w_1}:m_2,m_4,m_3,m_1,\mathbf{w_1}$\\
			$\SP^2_{m_1}:w_2,w_4,w_3,w_1,\mathbf{m_2}$&
			$\SP^2_{w_1}:m_3,m_4,m_1,m_2,\mathbf{w_2}$\\
			$\SP^2_{m_1}:w_3,w_4,w_1,w_2,\mathbf{m_3}$&
			$\SP^2_{w_1}:m_1,m_4,m_2,m_3,\mathbf{w_3}$\\
			$\SP^2_{m_1}:w_4,w_1,w_2,w_3,\mathbf{m_4}$&
			$\SP^2_{w_1}:m_4,m_1,m_2,m_3,\mathbf{w_4}$\\
			\hline
	\end{tabular} }	
	{}	
\end{table}	

The profile $\SP^2$ is identical to $\SP^1$ except $m_1, m_2, m_3, w_1, w_2,$ and $w_3$ move $w_4$ and $m_4$ up one position in their orderings. 
Again $\mu_1$, $\mu_2$ and $\mu_3$ from Table \ref{tab:NegStable} are stable with respect to $\SP^2$.  
However, $\mu_1$ and $\mu_3$ are the only egalitarian stable matchings. 
We claim that the putative profile $\PP^2$ in Table \ref{tab:PP2} is a (locally) minimally dishonest equilibrium for $\SP^2$ where $r(\PP^2)=\mu_2$ contradicting that $r(\cdot)$ always selects a sincere egalitarian stable matching. 

\begin{table}[ht]	
	\centering\caption	
	{Minimally Dishonest Equilibrium for $\SP^2$. \label{tab:PP2}}	
	{\begin{tabular}{|r l |r l |}	
			\hline
			$\PP^2_{m_1}:w_1,w_2,w_4,\mathbf{m_1},w_3$&
			$\PP^2_{w_1}:m_2,m_3,m_4,\mathbf{w_1},m_1$\\
			$\PP^2_{m_2}:w_2,w_3,w_4,w_1,\mathbf{m_2}$&
			$\PP^2_{w_2}:m_3,m_1,m_4,m_2,\mathbf{w_2}$\\
			$\PP^2_{m_3}:w_3,w_1,w_4,w_2,\mathbf{m_3}$&
			$\PP^2_{w_3}:m_1,m_2,m_4,m_3,\mathbf{w_3}$\\
			$\PP^2_{m_4}:w_4,w_1,w_2,w_3,\mathbf{m_4}$&
			$\PP^2_{w_4}:m_4,m_1,m_2,m_3,\mathbf{w_4}$\\
			\hline
	\end{tabular} }	
	{}	
\end{table}

The profile $\PP^2$ is again similar to $\PP^1$ except $m_2, m_3, w_2,$ and $w_3$ move $w_4$ and $m_4$ up one position in their ordering.
In addition, $m_1$ and $w_1$ move $w_4$ and $m_4$ in their preference lists. 
Similar to $\PP^1$ and $\SP^1$, only individual $m_1$ and $w_1$ are dishonest.  
The only matching stable with respect to $\PP^2$ is $\mu_2$ and therefore $r(\PP^2)=\mu_2$.
We now show $\PP^2$ is a minimally dishonest equilibrium for $\SP^2$ where $r(\PP^2)=\mu_2$, a matching that is not egalitarian with respect to $\SP^2$ thereby completing the proof of the theorem. 

First, $\PP^2$ is a Nash equilibrium:  No man can alter his preferences to get a partner that he prefers to his man-optimal partner \cite{GaleShapley62}.  With respect to the women's submitted preferences, each man is already receiving his man-optimal partner and therefore is providing a best response.  Symmetrically no woman can alter her preferences to receive a better outcome and $\PP^2$ is a Nash equilibrium.

It remains to show that each individual is (locally) minimally dishonest.  
Without loss of generality we examine only $m_1$.  
Man $m_1$ can only be more honest by submitting $\SP^2_{m_1}$.  
Next we will show $r([\PP^2_{-m_1}, \SP^2_{m_1}])=r([\PP^1_{-m_1}, \SP^1_{m_1}])$ implying that $\mu_3\in supp(r([\PP^2_{-m_1}, \SP^2_{m_1}]))$.  
This implies $m_1$ receives a strictly worse result if he is more honest. 
Therefore $\PP^2$ is a minimally dishonest equilibrium for $\SP^2$ that yields a non-egalitarian matching. 

To complete the proof, we show $r([\PP^2_{-m_1}, \SP^2_{m_1}])=r([\PP^1_{-m_1}, \SP^1_{m_1}])$. 
First, $r([\PP^2_{-m_1}, \SP^2_{m_1}])=r([\PP^2_{-m_1}, \SP^1_{m_1}])$:
Only $\mu_2$ and $\mu_3$ are stable with respect to $[\PP^2_{-m_1}, \SP^2_{m_1}]$ and $[\PP^2_{-m_1}, \SP^1_{m_1}]$.  
Therefore $p_{m_1y}([\PP^2_{-m_1}, \SP^2_{m_1}])=p_{m_1y}([\PP^2_{-m_1}, \SP^1_{m_1}])=0$ for all $y\notin \{\mu_2(m_1)=w_2, \mu_3(m_1)=w_3\}$.  
To transition from $\SP^2_{m_1}$ to $\SP^1_{m_1}$, $m_1$ must move $w_2$ up one position. 
Therefore by monotonicity, $p_{m_1w_2}([\PP^2_{-m_1}, \SP^1_{m_1}])\leq p_{m_1w_2}([\PP^2_{-m_1}, \SP^2_{m_1}])$.

Next,  INS implies $p_{m_1w_3}([\PP^2_{-m_1}, \SP^1_{m_1}])\leq p_{m_1w_3}([\PP^2_{-m_1}, \SP^2_{m_1}])$ since $p_{m_1w_4}([\PP^2_{-m_1}, \SP^1_{m_1}])=0$.
Therefore $p_{m_1y}([\PP^2_{-m_1}, \SP^2_{m_1}])= p_{m_1y}([\PP^2_{-m_1}, \SP^1_{m_1}])$ for all $y$.  
Since only $\mu_2$ and $\mu_3$ are stable and $\mu_2(m_1)\neq \mu_3(m_1)$, $r([\PP^2_{-m_1}, \SP^2_{m_1}])=r([\PP^1_{-m_1}, \SP^1_{m_1}])$. We can repeat this process for individuals $m_2, m_3, w_2$ and $w_3$ to show that  $r([\PP^2_{-m_1}, \SP^2_{m_1}])=r([\PP^1_{-\{m_1,w_1\}}, \SP^1_{m_1}, \PP^2_{w_1}])$. 
Moreover, INS implies the outcome is the same independent of which list from Table \ref{tab:P1alt} $w_1$ submits. Thus the process can be extended  to include $w_1$ so that  $r([\PP^2_{-m_1}, \SP^2_{m_1}])=r([\PP^1_{-m_1}, \SP^1_{m_1}])$. 
Thus $\PP_2$ is a (locally) minimally dishonest equilibrium for $\SP_2$ where $r(\PP_2)$ is not an egalitarian stable matching.
This is a contradiction, completing the proof of the theorem.
\end{proof}

\section{Fully-Randomized mechanisms}
\label{sec:Random}
Next, we examine Fully-Randomized stable matching mechanisms.  
Similar to Theorem \ref{thm:Main}, we show that stability is obtained when individuals are strategic and (locally) minimally honest.
In addition, we show that any sincerely stable matching can be obtained at a (locally)  minimally dishonest equilibrium.
Thus, similar to Theorem \ref{thm:Neg}, no Fully-Randomized stable matching mechanisms guarantees a sincere egalitarian stable matching. 

\begin{restatable}{theorem}{MainTwo}\label{thm:MGhasEQ} Suppose individuals are strategic and (locally) minimally dishonest.  Let $r(\cdot)$ be a Fully-Randomized stable matching mechanism.   Then  $\mu$ is sincerely stable with respect to $\SP$ if and only if there exists a (locally) minimally dishonest equilibrium $\PP$ for $\SP$ where $r(\PP)=\mu$.\end{restatable}

The second part of Theorem \ref{thm:MGhasEQ} follows almost identically to Theorem \ref{thm:Main}. 
Lemma \ref{lem:deterministic} guarantees that the outcome is selected deterministically at an equilibrium $\PP$.  
Therefore, $\PP$ has only one stable marriage since $r(\cdot)$ is Fully-Randomized and Corollary \ref{cor:one} extends to this setting.
We can immediately extend Corollary \ref{cor:CompleteOrder} to guarantee each individual is honest up to their spouse assigned by $r(\PP)$.  
The proof of sincere stability then immediately follows. 

Recall that every minimally dishonest Nash equilibrium is also locally minimally dishonest.
Thus, to show the first part of Theorem \ref{thm:MGhasEQ}, it is sufficient to show for any matching $\mu$ that is stable with respect to the sincere $\SP$, there is a minimally dishonest Nash equilibrium $\PP$ such that $r(\PP)=\mu$. 
To accomplish this, consider the preference profile $\PP^\mu$ obtained after each individual truncates their sincere profiles after their partner in $\mu$.  By construction, $\mu$ will be the only stable matching with respect to $\PP^{{\mu}}$.

We begin by presenting an algorithm to find a minimally dishonest equilibrium $\PP$ given $\PP^\mu$ where $r(\PP)=\mu$.  
Each iteration of the algorithm corrects a violation to minimal dishonesty and decreases the distance between an individual's putative and sincere preferences.  Since the distances are finite and integer, the algorithm terminates in finite time (Lemma \ref{lem:AlgorithmTerminates}).  
Next, we show that if $\PP$ is the profile obtained at the end of an iteration, then $\PP$ is an equilibrium and $\mu$ is the unique stable matching with respect to $\PP$ (Lemma \ref{lem:InvarientEquil}). 
This implies the algorithm outputs an equilibrium $\PP$ that yields the sincere stable matching $\mu$.  
We show that the equilibrium is minimally dishonest by showing that the algorithm does not terminate until all violations to minimal dishonesty are corrected (Lemma \ref{lem:InvarientDishonest}). Thus there is at least one minimally dishonest equilibrium $\PP$ that yields the matching $\mu$.

Algorithm \ref{Nash}, defined below,  finds a minimally dishonest equilibrium.
Let $Inv(r,{\SP},{\PP})$ be the set of violations to the minimally dishonest criterion given mechanism $r(\cdot)$, sincere profiles $\SP$, and equilibrium $\PP$. Formally, $Inv(r,\SP,\PP)$ is the set of $\{y,\PP'_y\}$ where $K({\PP}'_y, \SP_y)<K({\PP}_y, \SP_y)$ and $y$ obtains at least as good an outcome by submitting $\PP'_y$. 
Let $Inv'(r,{\SP},{\PP})\subseteq Inv(r,{\SP},{\PP})$ be the set of $\{y,{\PP}'_y\}\in Inv(r,{\SP},{\PP})$ where ${\PP}'_y$ agrees with $\SP_y$ up to $y$'s partner in $r(\PP)$.

We now have sufficient definitions to give an algorithm than finds a minimally dishonest equilibrium that yields an arbitrary sincere stable matching for Fully-Randomized mechanisms.  
Let $\mu$ be an arbitrary stable matching with respect to the sincere preferences. 
For each woman $w$, $m \PP_w m'$ if $m \SP_w m'$ for all $m,m'\in M$ and $w\PP_w m$ if $\mu(m) \SP_w m$ (i.e. she indicates she is only willing to match someone at least as good as her partner in $\mu$).
Equivalently, $w$ truncates $\SP_w$ after $\mu(w)$. 
Man $m$ also truncates $\SP_m$ after $\mu(m)$.
When given a Fully-Randomized algorithm $r(\cdot)$, sincere preference $\SP$ and the putative preferences $\PP^{\mu}$,  Algorithm \ref{Nash} will output a minimally dishonest equilibrium for $\SP$ that yields the sincerely stable matching $\mu$.

\begin{algorithm} \centering\caption{Equilibrium finding algorithm for Fully-Randomized and Gale-Shapley algorithms}\label{Nash}
	\begin{algorithmic}[1]
		\Procedure{EquilibriumFind}{}
		\While{$Inv'(r,\SP,\PP)\neq \emptyset$}
		\State Select $\{y,\PP'_y\}\in Inv'(r,\SP,\PP)$
		\State $\PP \gets [\PP_{-y},\PP'_y]$
		\EndWhile
		\State Output $\PP$
		\EndProcedure
	\end{algorithmic}
\end{algorithm}

We first show that the algorithm terminates regardless of the input.

\begin{lemma}Algorithm \ref{Nash}
	terminates.\label{lem:AlgorithmTerminates}\end{lemma}
\begin{proof}Consider the potential function
\begin{align}
\phi(\SP, \PP)=\sum_{y\in M\cup W} K(\SP_y,\PP_y).
\end{align}
By definition $\phi(\SP,\PP)\in \mathbb{Z}_{\geq 0}$ in each iteration of Algorithm \ref{Nash}. 
If individual $y$ updates their preferences in an iteration then $ K({\SP}_y, {\PP}_y)$ 
decreases by at least one and $K({\SP}_z, {\PP}_z)$ remains unchanged for $z\neq y$. 
Therefore Algorithm \ref{Nash} must terminate. 
Furthermore, $K(\SP_m,\PP_m)\leq {\binom{|W|+1}{2}}$ and 
$K(\SP_w, \PP_w)\leq {\binom{|M|+1}{2}}$ and Algorithm \ref{Nash} 
terminates in $O(|M||W|^2+|M|^2|W|)$ iterations.
\end{proof}

Next we show that at the end of each iteration, $\PP$ is a equilibrium and that $\mu$ is the only matching stable with respect to $\PP$. Thus, when Algorithm \ref{Nash} terminates, it outputs an equilibrium that yields the sincere stable matching $\mu$.

\begin{lemma}\label{lem:InvarientEquil} 
	Suppose Algorithm \ref{Nash} is given a Fully-Randomized mechanism $r(\cdot)$, $\SP$, and $\PP^\mu$ where $\mu$ is sincerely stable.  At the end of each iteration, $\PP$ is an equilibrium and $\mu$ is the only matching stable with respect to $\PP$.
\end{lemma}

\begin{proof}By construction of $\PP^\mu$, $\mu$ is the only stable matching at the beginning of the first iteration.  
Thus $\mu$ is both man and woman-optimal at the beginning of an iteration. 
No man can alter his preferences to obtain an outcome he prefers to the man-optimal matching \cite{GaleShapley62}. 
Therefore no man $m$ can alter $\PP^\mu_m$ to obtain a matching he prefers with respect to $\PP^\mu$.  
By construction of $\PP^\mu$, $m$ also cannot alter $\PP^\mu_m$ to obtain a matching he prefers with respect to $\SP_m$. 
Symmetrically, no woman can alter her preference to obtain a better outcome and $\PP^\mu$ is an equilibrium.

We now show that if $\PP$ is an equilibrium with the unique stable matching $\mu$ at the beginning of an iteration, then the condition holds at the end of the iteration completing the proof of the lemma.

Suppose that $\{y,\PP'_y\}\in Inv'(r, \SP, \PP)$ is updated in an iteration.  Let $\PP'=[\PP_{-y},\PP'_y]$ be the preference profile at the end of the iteration. Since $\PP$ is an equilibrium, $y$ cannot receive a better partner with respect to $\PP'$. Since $\{y,\PP'_y\}\in Inv'(r, \SP, \PP)$, $y$ does not obtain a worse partner with respect to $\PP'$ and $y$ is matched to $\mu(y)$ in every matching in $r(\PP')$.

We now claim that $\mu$ is the only matching stable with respect to $\PP'$. 
Suppose instead $\mu'\neq \mu$ is stable with respect to $\PP'$ but not $\PP$. Since $\mu'$ is not stable with respect to $\PP'$ there must be a pair $\{z,b\}$ that blocks $\mu'$ with respect to $\PP$ but not $\PP'$. 
Since only $y$ changed her preferences, $b=y$ and $\{z,y\}$ blocks $\mu'$ with respect to $\PP'$ implying $z\PP'_y \mu'(y)$ and $y\PP'_z \mu'(z)$.
If $\mu(y)=\mu'(y)$ then by construction of $Inv'(r,\SP,\PP)$, $z\PP_y \mu'(y)$ if and only if $z\PP'_y \mu'(y)$ and $\{z,y\}$ blocks $\mu'$ with respect to $\PP'$, a contradiction. 
If $\mu(y)\neq \mu'(y)$ then then $\mu(y)\PP_y \mu'(y)$ since $\PP$ is an equilibrium and $r(\cdot)$ being Fully-Randomized implies $y$ has positive probability of obtaining a worse outcome, contradicting the construction of $Inv'(r,\SP,\PP)$. 
Therefore $\mu$ is the only matching stable with respect to $\PP'$. 

It remains to show that $\PP'$ is still an equilibrium. This again follows directly from \cite{GaleShapley62} and the lemma holds.\end{proof}

By Lemmas \ref{lem:AlgorithmTerminates} and \ref{lem:InvarientEquil}, Algorithm \ref{Nash} outputs an equilibrium $\PP$ where $r(\PP)=\mu$.
It only remains to show that each individual is minimally dishonest. 
All individuals are minimally dishonest if and only if $Inv(r, \SP, \PP)=\emptyset$. Thus it suffices to show that in each iteration $Inv(r, \SP, \PP)=\emptyset$ if and only if $Inv'(r, \SP, \PP)=\emptyset$.

\begin{lemma}\label{lem:InvarientDishonest}Suppose Algorithm \ref{Nash} is given input $\SP$, $\PP^\mu$ and a Fully-Randomized $r(\cdot)$. In each iteration, $Inv(r, \SP, \PP)=\emptyset$ if and only if $Inv'(r, \SP, \PP)=\emptyset$.
\end{lemma}

\begin{proof} Since $Inv'(r,\SP,\PP)\subseteq Inv(r,\SP,\PP)$, the second direction holds immediately. Suppose that $\{y,\PP'_y\} \in Inv(r,\SP,\PP)$ and let $\PP'=[\PP_{-y},\PP'_y]$.

The first direction follows by selecting a  $\{y,\PP'_y\}\in Inv(r,\SP,\PP)$ and applying the proof techniques established in the second part of Theorem \ref{thm:MGhasEQ} and Corollaries \ref{cor:one} and \ref{cor:CompleteOrder} to show that we can modify $\{y,\PP'_y\}$ so that it is also in $Inv'(r,\SP,\PP)$.  
\end{proof}

%%%%%%%%%%%%%%%%%%%%% GALE SHAPLEY AND NEVER ONE SIDED

\section{The Gale-Shapley algorithm}\label{sec:sufficient}

While (locally) minimally dishonest equilibria always exist when using Fully-Randomized mechanisms, this does not necessarily hold for monotonic INS mechanisms.
In Proposition \ref{prop:NoEquil}, we give a monotonic INS mechanism and a set of preferences where no (locally) minimally dishonest equilibrium exists. 
We then show that the Gale-Shapley algorithm always has at least one minimally dishonest equilibrium and that every equilibrium yields the sincere woman-optimal matching. 

\begin{restatable}{proposition}{noequil}
	If $r(\cdot)$ selects an egalitarian stable matching uniformly at random with respect to the submitted preferences, then the there exists a sincere $\SP$ with no (locally) minimally dishonest equilibrium.\label{prop:NoEquil} 
\end{restatable}

\begin{proof}
Similar to Theorem \ref{thm:Neg}, we present a sincere profile $\SP$ and examine a (locally) minimally dishonest equilibrium $\PP$.
This mechanism is both monotonic and INS. 
Using these properties we are able to determine $\PP$ based on $\SP$.  
We then show there are two matchings stable with respect to $\PP$ contradicting Corollary \ref{cor:one} and therefore $\PP$ cannot exist. 
Consider the sincere preferences in Table \ref{tab:NoEquil}.

\begin{table}[ht]	
	\centering\caption	
	{Preferences $\SP$ with no Minimally Dishonest Equilibrium. \label{tab:NoEquil}}	
	{\begin{tabular}{|r l |r l |}	
			\hline
			${m_1}:$&$w_1,w_2,w_3,\mathbf{m_1}$&
			${w_1}:$&$m_3,m_1,m_2,\mathbf{w_1}$\\
			${m_2}:$&$w_2,w_3,w_1,\mathbf{m_2}$&
			${w_2}:$&$m_1,m_2,m_3,\mathbf{w_2}$\\
			${m_3}:$&$w_3,w_1,w_2,\mathbf{m_3}$&
			${w_3}:$&$m_2,m_3,m_1,\mathbf{w_3}$\\
			\hline
	\end{tabular} }	
	{}	
\end{table}

With respect to $\SP$, there are two stable matchings $\mu_1$ and $\mu_2$ where $\mu_1(m_1)=w_1$, $\mu_1(m_2)=w_2$, and $\mu_1(m_3)=w_3$; and $\mu_2(m_1)=w_2$, $\mu_2(m_2)=w_3$, and $\mu_2(m_3)=w_1$. 
Suppose $\PP$ is a minimally dishonest equilibrium.
By Theorem \ref{thm:Main}, $r(\PP)$ is either $\mu_1$ or $\mu_2$.  Without loss of generality suppose that $r(\PP)=\mu_1$. 

According to Corollary \ref{cor:CompleteOrder}, each man is honest about his first selection and each woman is honest about her first two selections.  Therefore $\PP$ is consistent with Table \ref{tab:NoEquil2}. 

\begin{table}[ht]	
	\centering\caption	
	{Beginning of $\PP$. \label{tab:NoEquil2}}	
	{\begin{tabular}{|r l |r l |}	
			\hline
			$\PP_{m_1}:$&$w_1,\mathbf{m_1}$&
			$\PP_{w_1}:$&$m_3,m_1,\mathbf{w_1}$\\
			$\PP_{m_2}:$&$w_2,\mathbf{m_2}$&
			$\PP_{w_2}:$&$m_1,m_2,\mathbf{w_2}$\\
			$\PP_{m_3}:$&$w_3,\mathbf{m_3}$&
			$\PP_{w_3}:$&$m_2,m_3,\mathbf{w_3}$\\
			\hline
	\end{tabular} }	
	{}	
\end{table}

First we claim that at least one man excludes his woman-optimal partner.  If not, then $\mu_2$ is stable with respect to $\PP$ contradicting Corollary \ref{cor:one}.  Without loss of generality we assume that this man is $m_1$ and $\PP$ is consistent with Table \ref{tab:NoEquil3}.  

\begin{table}[ht]	
	\centering\caption	
	{Updated $\PP$. \label{tab:NoEquil3}}	
	{\begin{tabular}{|r l |r l |}	
			\hline
			$\PP_{m_1}:$&$w_1,\mathbf{m_1},w_2$&
			$\PP_{w_1}:$&$m_3,m_1,\mathbf{w_1}$\\
			$\PP_{m_2}:$&$w_2,\mathbf{m_2}$&
			$\PP_{w_2}:$&$m_1,m_2,\mathbf{w_2}$\\
			$\PP_{m_3}:$&$w_3,\mathbf{m_3}$&
			$\PP_{w_3}:$&$m_2,m_3,\mathbf{w_3}$\\
			\hline
	\end{tabular} }	
	{}	
\end{table}

Similar to the proof of Theorem \ref{thm:Neg} we can conclude that $\mu_1$ is the woman-optimal matching regardless of how $m_2, m_3, w_1, w_2$ and $w_3$ fill out the remainder of their preferences.  Therefore by minimal dishonesty, everyone but $m_1$ is honest and the $\PP$ is consistent with Table \ref{tab:NoEquil4}. 
\begin{table}[ht]	
	\centering\caption	
	{Minimally Dishonest $\PP$.\label{tab:NoEquil4}}	
	{\begin{tabular}{|r l |r l |}	
			\hline
			$\PP_{m_1}:$&$w_1,\mathbf{m_1},w_2$&
			$\PP_{w_1}:$&$m_3,m_1,m_2,\mathbf{w_1}$\\
			$\PP_{m_2}:$&$w_2,w_3,w_1,\mathbf{m_2}$&
			$\PP_{w_2}:$&$m_1,m_2,m_3,\mathbf{w_2}$\\
			$\PP_{m_3}:$&$w_3,w_1,w_2,\mathbf{m_3}$&
			$\PP_{w_3}:$&$m_2,m_3,m_1,\mathbf{w_3}$\\
			\hline
	\end{tabular} }	
	{}	
\end{table}	

Moreover, $m_1$ submits either ($\PP_{m_1}:\ w_1, w_3, \mathbf{m_1}, w_2$) or ($\PP_{m_1}:\ w_1, \mathbf{m_1}, w_2, w_3$). 
In either case, $m_1$ can be more honest by swapping $\mathbf{m_1}$ with $w_2$ or $w_3$ respectively, instead submitting ($\PP_{m_1}': \  w_1, w_3, w_2, \mathbf{m_1}$).

Only $\mu_1$ and $\mu_2$ are stable with respect to $[\PP_{-m_1},\PP'_{m_1}]$. Moreover, $r([\PP_{-m_1},\PP'_{m_1}])= \mu_1$.  
Thus $m_1$ can get at least as good an outcome by submitting the more honest $\PP'_{m_1}$ contradicting that $m_1$ is (locally) minimally dishonest. 
\end{proof}

In this section, we also characterize the set of minimally dishonest equilibria obtained when $r(\cdot)$ is the Gale-Shapley algorithm. 
We prove that the sincere woman-optimal matching will always be obtained when individuals are (locally) minimally dishonest.

\begin{restatable}{theorem}{GSWom}\label{thm:GSFavorsWomen}Suppose individuals are strategic and (locally) minimally dishonest. If $r(\cdot)$ is the Gale-Shapley (man-optimal) algorithm, then (i) there always exists a (locally) minimally dishonest equilibrium and (ii) every minimally dishonest equilibrium yields the sincere women-optimal matching.\end{restatable}

\begin{proof} 
The proof of existence follows identically to Theorem \ref{thm:MGhasEQ}. 
Let $\PP_m=\SP_m$ for all $m$ and $\PP_w$ be obtained by truncating $\SP_w$ after $w$'s sincere woman optimal partner for all $w$. 
Only Lemma \ref{lem:InvarientEquil} specifically uses properties of the truly random mechanism. 
Specifically, Lemma \ref{lem:InvarientEquil} relies on $r(\cdot)$ selecting a matching that is not woman-optimal with positive probability. 
This property also holds for the Gale-Shapley (man-optimal) algorithm.
Therefore, when given the Gale-Shapley algorithm, $\SP$, and $\PP$ Algorithm \ref{Nash} outputs a (locally) minimally dishonest equilibrium $\PP$ where $r(\PP)$ is the sincere woman-optimal matching.

It remains to show that if $\PP$ is a (locally) minimally dishonest equilibrium then $r(\PP)$ is the sincere woman-optimal stable matching.  
The Gale-Shapley algorithm is monotonic and INS and thus any equilibrium must yield a sincerely stable matching. 
Suppose that $r(\PP)=\mu$ for a (locally) minimally dishonest equilibrium $\PP$.

We begin by showing that the (locally) minimally dishonest refinement implies that all men will be honest at every equilibrium.  
By Corollary \ref{cor:CompleteOrder}, $\PP_m$ agrees with $\SP_m$ up to $\mu(m)$ for man $m$.  
The Gale-Shapley algorithm only examines $\PP_m$ through $\mu(m)$, therefore the outcome is the same regardless of what appears after $\mu(m)$ in $\PP_m$.  
Therefore by (locally) minimal dishonesty, everything that appears after $\mu(m)$ in $\PP_m$ is consistent with $\SP_m$ and $m$ is completely honest.

Now since the men are honest and every woman is honest up to her partner in $\mu$, by Corollary \ref{cor:CompleteOrder},  the woman-optimal matching is stable with respect to the putative preferences.  
By Corollary \ref{cor:one}, there is only one matching stable with respect to the putative preferences and therefore $\mu$ is the woman-optimal matching.
\end{proof}

\section{Extensions}\label{sec:extensions}

In this section we examine extensions of SSM and minimal dishonesty.  We first consider the college admissions problem. In this setting we label the women as ``colleges" and men as ``students".  Unlike the stable matching problem, each college is allowed to match multiple students.  Each college has a quota indicating the maximum number of students that they are willing to match. We show that no stable matching mechanism can guarantee stability in the college admissions problem.

The student placement problem is a special case of the college admissions problem where each college is honest. 
Since the college admissions problem can be modeled by a stable matching problem by duplicating the colleges, the student placement problem is also a special case of the stable matching problem where all women are honest. 
We show that Theorem \ref{thm:Main} does not extend to this setting.
We then consider a more general model of the stable matching problem where any subset of individuals are allowed to be honest and show that the Gale-Shapley and Fully-Randomized algorithms still guarantee stability. 

We also examine coalitions of players and show that every minimally dishonest equilibrium is a strong equilibrium and therefore all our results extend when players are allowed to collude.
We then compare the minimally dishonest best response to a truncated best response.  
While the set of submitted preferences have similar properties in both settings, we establish that the two solution concepts are distinct.

We also consider the partial honesty and truncation refinements  discussed in Section \ref{sec:other}.
We show that the partial honesty refinement fails to guarantee stability when individuals are strategic.
While the truncation refinement guarantees stability, we show that the Nash equilibria are distinct from those obtained by minimal dishonesty.

Finally, we consider an alternative method of submitting preferences.  
Specifically, we consider how the equilibrium changes if a woman $w$ submits a total ordering on  the subset of $M$ she is willing to match (a truncation) instead of a total ordering on $M\cup \{w\}$.  
This alternative method of submitting preferences results in two new ways to measure honesty.
We show our results still hold in this setting.

\subsection{College admissions problem}

Roth has long claimed that the admissions problem is significantly different than the stable matching problem by showing that unlike the stable matching problem, a single college is able to alter their preferences to obtain a matching that they prefer to the college-optimal matching \cite{Roth85}. Very few of our results extend to the college admissions problem. 
We provide the unsettling result that  no stable matching mechanism can guarantee a sincere stable matching is selected at a (locally) minimally dishonest equilibrium and, like Roth, we must emphasize that the college admissions problem is different than the stable matching problem.

Without knowing colleges preferences between groups of students, we cannot cannot guarantee that a matching is deterministically selected at every equilibrium (Lemma \ref{lem:deterministic}).  Consider a set of preferences where a college has a $50\%$ chance of obtaining their 1st and 4th choice in students and  a $50\%$ chance of obtaining their 2nd and 3rd choice in students.  If the college is indifferent between these two outcomes, these preferences may correspond to an equilibrium. 

The most unsettling disparity between the stable matching problem and the college admissions problem is that we may obtain a matching at an equilibrium that is not sincerely stable, as shown in Theorem \ref{thm:CollegeAdmission}.  Mimicking Roth, Theorem \ref{thm:CollegeAdmission} demonstrates there is a set of preferences where a college can alter their preferences to obtain a matching that they prefer to the college-optimal matching even at a (locally) minimally dishonest equilibrium, regardless of the stable matching mechanism used. This implies that no stable matching mechanism  guarantees that a sincere stable matching is selected at every (locally) minimally dishonest equilibrium.

\begin{theorem}\label{thm:CollegeAdmission} When students and colleges are strategic, no stable matching mechanism always yields an equilibrium matching that is sincerely stable, with or without the (locally) minimal dishonesty refinement.
\end{theorem}

\begin{proof} The result without the refinement follows immediately from Lemma \ref{lem:deterministic} where each college has a capacity of one. Now consider the sincere preferences $\SP$ given in Table \ref{tab:College} where college $c_1$ has capacity for two students and colleges $c_2$ and $c_3$ have room for only one student.
These preferences were originally presented in \cite{Roth85} to show that a college can manipulate its preferences to obtain an outcome superior to the college-optimal stable matching.
The only stable matching with respect to these preferences is $\mu$ where $\mu(c_1)=\{s_3,s_4\}$, $\mu(c_2)=s_2$, and $\mu(c_3)=s_1$.

\begin{table}[ht]	
	\centering\caption	
	{Preferences $\SP$ for Theorem \ref{thm:CollegeAdmission}. \label{tab:College}}	
	{\begin{tabular}{|r l |r l |}	
			\hline
			$\SP_{c_1}:$&$s_1,s_2,s_3,s_4,\mathbf{c_1}$&
			$\SP_{s_1}:$&$c_3,c_1,c_2,\mathbf{s_1}$\\
			$\SP_{c_2}:$&$s_1,s_2,s_3,s_4,\mathbf{c_2}$&
			$\SP_{s_2}:$&$c_2,c_1,c_3,\mathbf{s_2}$\\
			$\SP_{c_3}:$&$s_3,s_1,s_2,s_4,\mathbf{c_3}$&
			$\SP_{s_3}:$&$c_1,c_3,c_2,\mathbf{s_3}$\\
			& &
			$\SP_{s_4}:$&$c_1,c_2,c_3,\mathbf{s_4}$\\
			\hline
	\end{tabular} }	
	{}	
\end{table}	

For contradiction, suppose there is a mechanism $r(\cdot)$ and a (locally) minimally dishonest equilibrium $\PP$ where $r(\PP)=\mu$. 
We now that show there are enough completely honest colleges and students to guarantee that $c_1$ is able to alter $\PP_{c_1}$ to obtain a matching that $c_1$ prefers to $\mu$, contradicting that $\PP$ is an equilibrium.

Similar to Theorem \ref{thm:Neg} and Proposition \ref{prop:NoEquil} we gradually reveal the structure of $\PP$ until we are able to show that it is not an equilibrium.
We begin by showing that there is only one matching stable with respect to $\PP$.  
This allows us to extend Corollary \ref{cor:CompleteOrder} to this setting for every individual that is allowed only one spouse (i.e. everyone but $c_1$ is guaranteed to be honest up their $\mu$-partner).  
This implies that everyone will be completely honest ($\SP=\PP$). 
However, we show that $\SP$ is not a Nash equilibrium, a contradiction. 

First we show that $\mu$ is the only matching stable with respect to $\PP$:  Let $\mu'$ be a matching stable with respect to $\PP$.  By the Rural Hospital Theorem, $c_1$ has two partners in $\mu'$.  If $c_1$ submits $s_i \PP'_{c_1} c_1$ if and only if $s_i\in \mu'(c_1)$ then $c_1$ is assigned $\mu'(c_1)$.  $\PP$ is a Nash equilibrium implying $c_1$ cannot strictly prefer this outcome and therefore $\mu(c_1)=\mu'(c_1)$.  This implies $\mu'(c_2)\in \{s_1,s_2\}$.  Through the same reasoning, $\mu(c_2)=\mu'(c_2)$ and $\mu=\mu'$ completing the first claim. 

Next, colleges $c_2$ and $c_3$ and all students are honest up to their $\mu$-partner: This follows in the same fashion as Corollary \ref{cor:CompleteOrder}.  The proof of Corollary \ref{cor:CompleteOrder} only relies on (i) each individual is allowed a single spouse and (ii) there is a single stable matching with respect to $\PP$. In Section \ref{sec:Mono}, we establish (ii) by using either the INS and monotonicity properties or the Fully-Randomized property.  However, we were able to establish this in the previous claim without either of these properties.  Furthermore, colleges $c_2$ and $c_3$ and all students are allowed a single spouse completing the second claim. 

We now show $c_1$ is completely honest.  
Suppose instead $\PP_{c_1}\neq \SP_{c_1}$.
Since $\mu$ matches $c_1$ to $\{s_3,s_4\}$, $s_3,s_4 \PP_{c_1} \mathbf{c_1}$. 
However, since $\PP_{c_1}\neq \SP_{c_1}$, at least one of the following holds: (i) $s_4\PP_{c_1}s_3$, (ii) $\mathbf{c_1}\PP_{c_1}s_1$ (iii) $\mathbf{c_1}\PP_{c_1}s_2$, (iv) $s_2 \PP_{c_1}s_1$, or (v) $s_3\PP_{c_1}s_2\PP_{c_1}\mathbf{c_1}$.
Consider case (i) and suppose $c_1$ swaps $s_4$ and $s_3$ in $\PP_{c_1}$. 
By the previous claim, each of the students honestly report their first choice and therefore $\mu$ remains stable. 
Therefore, by the Rural Hospital theorem $c_1$ is still matched with two students after swapping $s_4$ and $s_3$.  
Since $s_3$ and $s_4$ are his least preferred students, he obtains at least as good of an outcome after becoming more honest, a contradiction to (locally) minimal dishonesty. 
Therefore $s_3\PP_{c_1}s_4$.
An identical argument works for the four remaining cases and therefore $\PP_{c_1}=\SP_{c_1}$. 

%%%%%%%%%%%%%%%%%%%%%%%

Next, we show that everyone else is honest. 
Similar to the proof of Theorem \ref{thm:Neg} we can complete the Gale-Shapley algorithm on $\PP$ without knowing the remaining details of $\PP$.  Using only the details we have available ($c_1$ is honest and everyone else is honest up to their spouse), we are able to conclude that $\mu$ is the only matching stable regardless of how each individual fills out their remaining preferences.  Thus, by (locally) minimal dishonesty, each of these individuals are honest and the minimally dishonest equilibrium is the sincere profile $\SP$. 

However, $\PP=\SP$ is not an equilibrium since college $c_1$ can update his preferences to ($\PP'_{c_1}: s_1,  s_4, \mathbf{c_1}, s_2, s_3$) in order to obtain the matching $\mu'$ where $\mu'(c_1)=\{s_1,s_4\}$, $\mu'(c_2)=s_2$, $\mu'(c_3)=s_3$, a matching that $c_1$ prefers to $\mu$.  This contradicts that $\PP$ is an equilibrium. Therefore there is no mechanism that guarantees a sincere stable matching for the college admissions problem.\end{proof}

\subsection{Truth-tellers and the student placement problem} 

In practice, some individuals may prefer to be honest regardless of whether they can manipulate their preferences to obtain a partner they strictly prefer.  Such individuals are often called  \emph{truth-tellers}, and there is experimental evidence that they exist \cite{Gneezy2005,WouldILie}. 
The student placement problem is an instance of the stable matching (college admissions) problem where all women (colleges) are truth-tellers and where all the men (applicants) are strategic.

Lemma \ref{lem:deterministic} does not hold in this setting and a matching is not always selected deterministically at an equilibrium. If all individuals are truth-tellers and the decision mechanisms selects a matching uniformly at random, then at the unique equilibrium ${\PP}=\SP$ every sincere stable matching has positive probability of being selected. As a result, we do not treat $r({\PP})$ as a singleton in this section. However, through the same proof technique in Lemma \ref{lem:deterministic}, we can prove  that if an individual is strategic then their partner is selected deterministically at every equilibrium.

Corollary \ref{cor:one} also does not necessarily hold in this setting and there may be more than one matching stable with respect to the equilibrium preferences. 
As a result we cannot guarantee that a strategic individual is honest up to their spouse (Corollary \ref{cor:CompleteOrder}). 
Without Corollary \ref{cor:CompleteOrder}, we cannot extend Theorem \ref{thm:Main} to the student placement problem:

\begin{proposition}
	Suppose each individual is either a truth-teller or strategic and (locally) minimally dishonest.
	If $r(\cdot)$ selects an egalitarian stable matching uniformly at random with respect to the submitted preferences, then the outcome SSM may be sincerely unstable.
	\label{prop:sadface}
\end{proposition}

\begin{proof}
The sincere preferences are given in Table \ref{tab:Contra}. 
The stable matchings with respect to $\SP$ are $\mu_1$ and $\mu_2$ from Table \ref{tab:ContraStab}.
The egalitarian cost of matching is found by charging $i$ points if an individual is matched to his or her $i$th choice. 
An egalitarian matching is found by selecting the stable matching with the lowest cost.
The only stable marriages are $\mu_1$ and $\mu_2$ and both have an egalitarian cost of 14. 
Therefore $r(\SP)$ selects between $\mu_1$ and $\mu_2$ uniformly at random.  

\begin{table}[ht]	
	\centering\caption	
	{Preferences $\SP$ for Proposition \ref{prop:sadface} \label{tab:Contra}}	
	{\begin{tabular}{|r l |r l |}	
			\hline
			$\SP_{m_1}:$&$w_2,w_1,\mathbf{m_1},w_3, w_4$&
			$\SP_{w_1}:$&$m_1,m_3,m_2,\mathbf{w_1},m_4$\\
			$\SP_{m_2}:$&$w_1,w_3,w_2,\mathbf{m_2}, w_4$&
			$\SP_{w_2}:$&$m_2,m_1,\mathbf{w_2},m_3,m_4$\\
			$\SP_{m_3}:$&$w_4,w_3,\mathbf{m_3},w_1,w_2$&
			$\SP_{w_3}:$&$m_2,m_3,m_4,\mathbf{w_3},m_1$\\
			$\SP_{m_4}:$&$w_3,w_4,\mathbf{m_4},w_1, w_2$&
			$\SP_{w_4}:$&$m_4,m_3,\mathbf{w_4},m_1,m_2$\\
			\hline
	\end{tabular} }	
	{}	
\end{table}	

\begin{table}[ht]	
	\centering\caption	
	{Matchings for Proposition \ref{prop:sadface} \label{tab:ContraStab}}	
	{\begin{tabular}{|r|l |}	
			\hline
			$\mu_1$ & $\mu_1(m_1)=w_2, \mu_1(m_2)=w_1, \mu_1(m_3)=w_4, \mu_1(m_4)=w_3$\\ 
			$\mu_2$ & $\mu_2(m_1)=w_2, \mu_2(m_2)=w_1, \mu_2(m_3)=w_3, \mu_2(m_4)=w_4$\\ 
			$\mu_3$ & $\mu_3(m_1)=w_1, \mu_3(m_2)=w_2, \mu_3(m_3)=w_3, \mu_3(m_4)=w_4$\\ 
			\hline
	\end{tabular} }	
	{}	
\end{table}

Suppose now that $w_3$ is strategic and that all men are truth-tellers.  
Woman $w_3$ can alter her preferences to obtain her woman-optimal partner, $m_3$, with probability one. 
Suppose instead that $w_3$ submits ($\PP_{w_3}: \ m_3, m_2, m_4, \mathbf{w_3}, m_1$). 
With respect to $[\SP_{-w_3},\PP_{w_3}]$ only $\mu_1$ and $\mu_3$ are stable. 
Moreover, the egalitarian cost of $\mu_1$ remains $14$ while the egalitarian cost of $\mu_3$ with respect to $[\SP_{-w_3},\PP_{w_3}]$ is $13$. 
Therefore $\mu_3$ is selected with probability one and $w_3$ is assigned her woman-optimal partner $m_3$. 

It is well known that $w_3$ cannot alter her preferences to obtain anything better than her woman-optimal partner therefore she is providing a best response.  
Moreover, she is (locally) minimally dishonest since $K(\PP_{w_3},\SP_{w_3})=1$. 
No other woman can alter her preferences to get a better outcome since all other women are receiving their first choice. 
No man can alter his preferences since all men are truth-tellers and therefore $[\SP_{-w_3},\PP_{w_3}]$ is a minimally dishonest equilibrium.
The pair $\{m_2,w_3\}$ blocks the matching $\mu_3$ with respect to $\SP$ and therefore $r([\SP_{-w_3},\PP_{w_3}])=\mu_3$ is unstable with respect to $\SP$ completing the proof of the proposition.
\end{proof}

Despite our main result not extending to this setting, we can still give a class of algorithms that guarantees stability when there are truth-tellers and minimally dishonest individuals. Specifically, the Fully-Randomized and Gale-Shapley algorithms still guarantee stability. 

\begin{restatable}{theorem}{MG}\label{thm:MG2}Suppose each individual is either a truth-teller or strategic and minimally dishonest.  If $r(\cdot)$ is Fully-Randomized, then (i) there there always exists a (locally) minimally dishonest equilibrium and (ii) for every (locally) minimally dishonest equilibrium ${\PP}$, $\mu$ is sincerely stable for all $\mu\in r({\PP})$.\end{restatable}

\begin{restatable}{theorem}{GaleTwo}\label{thm:Gale2}Suppose each individual is either a truth-teller or strategic and (locally) minimally dishonest.  If $r(\cdot)$ is the Gale-Shapley (man-optimal) algorithm, then (i) there always exists a (locally) minimally dishonest equilibrium and (ii) every (locally) minimally dishonest equilibrium assigns every strategic woman her sincere women-optimal partner.\end{restatable}

The proofs of Theorems \ref{thm:MG2}  and \ref{thm:Gale2} follow identically to Theorems \ref{thm:MGhasEQ} and \ref{thm:GSFavorsWomen} after we reestablish Corollary \ref{cor:CompleteOrder} for the  Fully-Randomized and Gale-Shapley algorithms.

\begin{lemma}
	Let $r(\cdot)$, $\SP$ and $\PP$ be as in the statement of Theorem \ref{thm:MG2} or \ref{thm:Gale2} . If $z$ is strategic then for all $\mu(z)=\mu'(z)$ for all $\mu$ and $\mu'$ stable with respect to $\PP$. \label{lem:one2}
\end{lemma}

\begin{proof}
Without loss of generality, suppose $z=w$ is a woman. 
If $r(\cdot)$ always selects the woman-optimal matching then the best response for $w$ is to be honest.  
Therefore the statement of the lemma holds for the Gale-Shapley algorithm when women propose. 
If the statement of the lemma fails to hold for a Fully-Randomized algorithm, then there are at least two stable marriages with respect to $\PP$ implying that that are at least two stable matchings in $supp(r(\PP))$ contradicting Lemma \ref{lem:deterministic}. 
\end{proof}

\begin{corollary}
	Let $r(\cdot)$, $\SP$ and $\PP$ be as in the statement of \ref{thm:MG2} or Theorem \ref{thm:Gale2}  and let $\mu\in r(\PP)$. If $y_1 \SP_z \mu(z)$ and $y_1 \SP_z y_2$, then $y_1 \PP y_2$. \label{cor:CompleteOrder2}
\end{corollary}

Corollary \ref{cor:CompleteOrder2} follows in the same fashion as Corollary \ref{cor:CompleteOrder}. Theorems \ref{thm:MG2} and \ref{thm:Gale2}  immediately follow.

\subsection{Coalitions}

Gale and Sotomayor also specifically motivate the study of manipulation when collusion is allowed \cite{GusfieldIrving89}. In this section we consider coalitions and \emph{strong equilibria} -- equilibria where no group of individuals can collude such that every member of the group obtains a strictly better outcome. We show that for any monotonic INS stable matching mechanism, every (locally) minimally dishonest equilibrium is also a strong equilibrium. This implies that all the results from previous sections apply even when collusion is allowed. In addition, it implies that the core of the SSM game is non-empty when using the Gale-Shapley algorithm or a monotonic INS Fully-Randomized stable matching mechanism even when we refine the set of equilibria to those where everyone is minimally dishonest.

\begin{theorem}\label{thm:MDImplyStrong}Let $r(\cdot)$ be an arbitrary monotonic INS or Fully-Randomized stable matching mechanism.  Every minimally dishonest equilibrium is a strong equilibrium.\end{theorem}

\begin{proof}By Corollary \ref{cor:one}, $r(\PP)=\mu$ is the unique putatively stable matching and therefore is putatively both man and woman-optimal.  By \cite{Demange87}, no coalition of men and women can alter $\PP$ such that every one of them prefers the outcome to $\mu$. By Corollary \ref{cor:CompleteOrder}, each individual prefers $\mu$ with respect to $\PP$ if and only if he/she prefers $\mu$ with respect to $\SP$. Therefore no coalition can alter $\PP$ to obtain a matching they all sincerely prefer to $\mu$.  Therefore $\PP$ is a strong equilibrium. \end{proof}

The converse does not necessarily hold, as discussed previously.
For instance, Gale and Sotomayor demonstrated that there are strong equilibria for the Gale-Shapley algorithm that yield stable matchings that are not woman-optimal \cite{GaleSotomayor85}. However, in Theorem \ref{thm:GSFavorsWomen}, we established that all (locally) minimally dishonest equilibria yield the sincere woman-optimal matching.

Every strong equilibrium of SSM results in a sincerely stable matching \cite{Shin96,Sonmez97}, therefore Theorem \ref{thm:MDImplyStrong} provides alternative proofs of Theorems  \ref{thm:Main} and \ref{thm:MGhasEQ}.
Since every (locally) minimally dishonest equilibrium of a monotonic INS or Fully-Randomized $r(\cdot)$ is a strong equilibrium and since every strong equilibrium yields a sincere stable matching, every (locally) minimally dishonest equilibrium yields a sincere stable matching.

If a mechanism has a (locally) minimally dishonest equilibrium, then it also has a strong equilibrium and the core is non-empty.  For instance, when using the Gale-Shapley algorithm or a Fully-Randomized $r(\cdot)$, the SSM game with the (locally) minimally dishonest refinement has a non-empty core.

\begin{corollary}\label{thm:NonEmptyCore} Let $\SP$ be arbitrary and $r(\cdot)$ be the Gale-Shapley Algorithm or a Fully-Randomized stable matching mechanism. The SSM game with $r(\cdot)$ and the (locally) minimally dishonest refinement has a non-empty core.\end{corollary}

Corollary \ref{thm:NonEmptyCore} follows immediately from Theorems \ref{thm:MGhasEQ}, \ref{thm:GSFavorsWomen},  and \cite{Shin96,Sonmez97}.  
\subsection{Distinction between partial honesty and minimally dishonesty}\label{sec:partialhonest}

Partial honesty is another honesty refinement proposed in the voting literature.  It requires the individual to be completely honest unless their dishonesty positively impacts their valuation of the outcome. 

\begin{definition}\label{def:PartialHonesty}  Let $\Pi$  be the sincere preferences and let $\SPi$ be an equilibrium in the Strategic Stable Matching game where $r(\SPi)=\mu_1$.   Individual  $y$ is \emph{partially honest} if $\mu_1(y) \SP_y \mu_2(y)$ for some $\mu_2(y)\in supp\left(r\left([\SPi_{-y},\pi_y]\right)\right)$.
\end{definition}

A Nash equilibrium is a partially honest Nash equilibrium if each individual is partially honest.
This condition requires that individual view honesty in a binary fashion. 
This is in contrast to the experimental evidence from Section \ref{sec:averse} that suggests individuals have a more nuanced view of honesty.
While partial honesty removes some of the absurd equilibria from Corollary \ref{thm:AllIndividuallyRational}, we show that it fails to eliminate them all.

\begin{proposition}\label{prop:PartialHonesty}
	There exists a monotonic, INS, Fully-Randomized stable matching mechanism $r$, sincere $\Pi$, a partially honest Nash equilibrium $\SPi$ where $r(\SPi)$ is not stable with respect to $\Pi$. 
\end{proposition}

\begin{proof}
Let $r$ select a stable matching uniformly at random. 
Consider the sincere preferences in Table \ref{tab:Partial Honesty}.
\begin{table}[ht]	
	\centering\caption	
	{Preferences $\SP$ for Proposition \ref{prop:PartialHonesty}.\label{tab:Partial Honesty}}	
	{\begin{tabular}{|c|c|}	
			\hline
			$\SP_{m_1}: \ w_1,w_2,w_3,\mathbf{m_1},w_4$&
			$\SP_{w_1}: \ m_1,m_2,m_3,\mathbf{w_1},m_4$\\
			$\SP_{m_2}: \ w_4,w_1,\mathbf{m_2},w_2,w_3$&
			$\SP_{w_2}: \ m_4,m_1,\mathbf{w_2},m_2,m_3$\\
			$\SP_{m_3}: \ w_1,w_4,\mathbf{m_3},w_2,w_3$&
			$\SP_{w_3}: \ m_1,m_4,\mathbf{w_3},m_2,m_3$\\
			$\SP_{m_4}: \ w_3,w_2,\mathbf{m_3},w_1,w_4$&
			$\SP_{w_4}: \ m_3,m_2,\mathbf{w_3},m_1,m_4$\\
			\hline
	\end{tabular}}
	{}	
\end{table}	

\begin{table}[ht]	
	\centering\caption	
	{Matchings for Proposition \ref{prop:PartialHonesty} \label{tab:PartialHonesty}}	
	{\begin{tabular}{|r|l |}	
			\hline
			$\mu_1$ & $\mu_1(m_1)=w_1, \mu_1(m_3)=w_4, \mu_1(m_4)=w_3$\\ 
			$\mu_2$ & $\mu_2(m_1)=w_2, \mu_2(m_2)=w_1, \mu_2(m_3)=w_4, \mu_2(m_4)=w_3$\\ 
			$\mu_3$ & $\mu_3(m_1)=w_3, \mu_3(m_2)=w_2, \mu_3(m_3)=w_4, \mu_3(m_4)=w_2$\\ 
			\hline
	\end{tabular} }	
	{}	
\end{table}	

With respect to these preferences, the only stable matching is $\mu_1$ from Table \ref{tab:PartialHonesty}. 
Suppose instead that $m_1$ submits $(\SPi_{m_1}: w_4, w_2, \mathbf{m_1}, w_1, w_3)$, that $w_1$ submits $(\SPi_{w_1}: m_4, m_2, \mathbf{w_1}, m_1, m_3)$ and that all other individuals submits $\SPi_y=\Pi_y$.
With respect to these submitted preferences, only $\mu_2$ is stable and therefore $r(\SPi)=\mu_2$. 
As in earlier proofs, it is straightforward to verify that $\SPi$ is a Nash equilibrium. 
Moreover, it is a partially honesty Nash equilibrium;
if $m_1$ instead submits his honest $\Pi_m$ then both $\mu_2$ and $\mu_3$ are stable and $m_1$ obtains a strictly worse outcome. 
Symmetrically, $w_1$ is partially honest completing the proof of the proposition. 
\end{proof}

The preferences given in the proof of Proposition \ref{prop:PartialHonesty} are not absurd because they fail to result in a stable outcome. 
Rather the preferences are unreasonable because of their relationship to the sincere preferences:
partial honesty indicates that it is perfectly reasonable for an individual to indicate they are willing to match someone that they have no interest in. 
Moreover, partial honesty allows $m_1$ to refuse to match his first choice $w_1$ despite $w_1$ most preferring $m_1$.  
Thus we view partial honesty as unfit for describing behavior in the setting of stable matchings. 

\subsection{Distinction between minimal truncation and minimally dishonest}

Another refinement proposed in \cite{Roth99, Ehlers08} is the truncation refinement -- a restriction that only permit man $m$ to submit $\PP_m$ if $w \SP_m w'$ implies $w \PP_m w'$ for all $w,w'\in W$. 
This restriction does not allow individuals to permute their preferences, but it does allow them to declare potential spouses as unacceptable.
A minimally dishonest equilibrium and a minimally truncated equilibrium have similar properties. 
It is straightforward to show that if individuals provide minimally truncated best responses then the corresponding set of equilibria also yield sincerely stable matchings. 
Moreover, Algorithm \ref{Nash} will sometimes output minimally truncated best responses.  
In this section, we establish that the two response functions are distinct. 

\begin{proposition}
	A minimally truncated equilibrium and a (locally) minimally dishonest equilibrium are distinct concepts. \label{prop:distinct}
\end{proposition}

\begin{proof} With respect to the preferences in Table \ref{tab:SincereTrunc} there are two stable matchings $\mu_1$ and $\mu_2$ where $\mu_1(m_1)=w_2, \mu_1(m_2)=w_1$, and $\mu_1(m_3)=w_3$; and $\mu_2(m_i)=w_i$ for all $i$.  
The Gale-Shapley algorithm implies the man-optimal matching $\mu_1$ is selected. 
Woman $w_1$ can adjust her preferences to obtain her woman-optimal partner $\mu_2(w_1)=m_1$.  
Using truncation, $w_1$ can only do this by submitting the list ($\PP^1_{w_1}: m_1, \mathbf{w_1}, m_2, m_3$).
When everyone else is honest this corresponds to a minimally truncated equilibrium.
She can also obtain $m_2$ without truncation by submitting ($\PP^1_{w_1}:  m_1, m_3, m_2, \mathbf{w_1}$). 
Moreover, $d(\SP_{w_1},\PP^2_{w_1})=1$ and therefore $[\SP_{-w_1},\PP^1_{w_1}]$ is a (locally) minimally dishonest equilibrium.  
However, $d(\SP_{w_1},\PP^1_{w_1})=2$ and the truncated best response is not a (locally) minimally dishonest best response.  
Therefore the two strategy concepts are distinct.  \end{proof}

\begin{table}[ht]	
	\centering\caption	
	{Preferences $\SP$ for Proposition \ref{prop:distinct}. \label{tab:SincereTrunc}}	
	{\begin{tabular}{|r l |r l |}	
			\hline
			${\SP}_{m_1}:$&$w_2,w_1,\mathbf{m_1},w_3$&
			${\SP}_{w_1}:$&$m_1,m_2,m_3,\mathbf{w_1}$\\
			${\SP}_{m_2}:$&$w_1,w_2,\mathbf{m_2},w_3$&
			${\SP}_{w_2}:$&$m_2,m_3,m_1,\mathbf{w_2}$\\
			${\SP}_{m_3}:$&$w_1,w_3,\mathbf{m_3},w_2$&
			${\SP}_{w_3}:$&$m_3,m_1,m_2,\mathbf{w_3}$\\
			\hline
	\end{tabular} }	
	{}	
\end{table}

\subsection{Evaluating honesty with truncated preference lists}

If $\mathbf{w}\PP_w m$ then there is no stable matching $\mu$ where $\mu(w)=m$.  
As such, the ordering of the men $\{m: \mathbf{w}\PP_w m\}$ is irrelevant when evaluating the stability of a matching.
For this reason, there is some disparity in the literature as to whether woman $w$ submits a preference list that is a total ordering on $M\cup\{w\}$ or she submits a total ordering on a subset of $M$ (i.e. an incomplete or truncated list). 
This distinction does not alter the decision process for any algorithm commonly referenced in the literature. 
However, it does alter how an individual measures honesty. 
As such, we also considered preferences as truncated lists. 

Most of our results our written so that they only rely on how an individual ranks a spouse they are willing to match.  
As such, these results immediately extend when individuals evaluate honesty with truncated lists.
Several of our proofs however utilize the full structure of the preferences.  
For these results, we have written the proofs so that it is straightforward to adjust when individuals evaluate honesty with truncated lists. 
For instance in Table \ref{tab:P1alt} of Theorem \ref{thm:Neg} we provide two options for $w_1$'s submitted preferences.  
If individuals evaluate honesty with truncated lists then she must submit the first option to be minimally dishonest.

We denote a truncated list with (${\PP_m}:w_1, w_2, \mathbf{m}, \{w_3, w_4\}$) to indicate that $m$ is unwilling to match either $w_3$ or $w_4$.  
Let $D(\Pi^1_m,\Pi^2_m)$ be the set of unordered $\{i,j\}\in W\cup\{m\}$ where $i$ and $j$ appear in a different order in $\Pi^1_m$ and $\Pi^2_m$. 
Let ${\cal R}_k(\SP^1_m,\SP^2_m)$ be the set of $\{i,j\}\in W\cup\{m\}$ where $m$ is unwilling to match $i$ and $j$ in $\SP^k_m$ but prefers $i$ to $j$ in $\SP^{3-k}_m$.
With this definition we present the two common generalizations of the Kendall Tau distance.

\begin{definition}[Kendall Tau with Penalty $p$]\label{def:metric2}
	Let $\Pi^1_m$ and $\Pi^2_m$  be two truncated preference lists over a set and let $p\in [0,1]$.  Then the {\it Kendall Tau distance with penalty $p$} between $\Pi^1_m$ and $\Pi^2_m$ is
	\begin{align}\label{eqn:metric}
	K^{(p)}(\Pi^1_m,\Pi^2_m) = |D(\Pi^1_m,\Pi^2_m)|+p\left(|{\cal R}_1(\Pi^1_m,\Pi^2_m)|+|{\cal R}_2(\Pi^1_m,\Pi^2_m)|\right).
	\end{align}
\end{definition}

\begin{definition}[Hausdorff Distance Based on Kendall Tau]\label{def:metric3}
	Let $\Pi^1_m$ and $\Pi^2_m$  be two truncated preference lists over a set and let $p\in [0,1]$.  Then the {\it Hausdorff distance based on Kendall Tau} between $\Pi^1_m$ and $\Pi^2_m$ is
	\begin{align}\label{eqn:metric2}
	K_{Haus}(\Pi^1_m,\Pi^2_m) = |D(\Pi^1_m,\Pi^2_m)|+max\left\{|{\cal R}_1(\Pi^1_m,\Pi^2_m)|,|{\cal R}_2(\Pi^1_m,\Pi^2_m)|\right\}.
	\end{align}
\end{definition}

When $r(\cdot)$ is monotonic and INS it is straightforward to show that if $\mathbf{z} \SP_z y$ then $\mathbf{z} \PP_z y$ at a (locally) minimally dishonest equilibrium. 
Therefore ${\cal R}_2(\PP_y,\SP_y)=0$ and $K_{Haus}(\PP_y,\SP_y)=K^{(1)}(\PP_y,\SP_y)$ at a minimally dishonest equilibrium. 

Moreover the metric $K(\cdot)$ when individuals are using a total ordering is equivalent to the metric $K^{(0)}(\cdot)$ if individuals use the equivalent truncated list:
Let $\PP_w$ be a total ordering on $M\cup\{w\}$ and let $\PP'_w$ be the equivalent truncated list. 
When working with totally ordered preference lists, (locally) minimal dishonesty implies (i) if $\mathbf{w} \SP_w m$ then $\mathbf{w} \PP_w m$ and (ii) if $m_i \SP_w m_j$, $\mathbf{w} \PP_w m_i$, and $\mathbf{w}\PP_w m_j$ then $m_i \PP_w m_j$.   
As a result, $K(\PP_w,\SP_w)=K^{(0)}(\PP'_w,\SP_w)$. 
Therefore when $p=0$, we obtain the same set of (locally) minimally dishonest equilibria whether individuals use total orderings or truncated lists.  
However, the set of equilibria may be different when $p>0$. 

\begin{proposition}\label{prop:DiffHonest}{Truncated lists may result in different minimally dishonest equilibria when $p>0$.}\end{proposition}

\begin{proof}
We consider SSM using the Gale-Shapley algorithm with the sincere profile in Table \ref{tab:Trunc}.  The stable matchings are given in Table \ref{tab:TruncMar}. 

\begin{table}[ht]	
	\centering\caption	
	{Preferences $\SP$ for Proposition  \ref{prop:DiffHonest}. \label{tab:Trunc}}	
	{\begin{tabular}{|r l |r l |}	
			\hline
			$\SP_{m_1}:$&$w_1,w_2,w_3,w_4,\mathbf{m_1}$&
			$\SP_{w_1}:$&$m_2,m_3,m_1,m_4,\mathbf{w_1}$\\
			$\SP_{m_2}:$&$w_2,w_3,w_1,w_4,\mathbf{m_2}$&
			$\SP_{w_2}:$&$m_3,m_1,m_2,m_4,\mathbf{w_2}$\\
			$\SP_{m_3}:$&$w_3,w_1,w_2,w_4,\mathbf{m_3}$&
			$\SP_{w_3}:$&$m_1,m_2,m_3,m_4,\mathbf{w_3}$\\
			$\SP_{m_4}:$&$w_4,w_1,w_2,w_3,\mathbf{m_4}$&
			$\SP_{w_4}:$&$m_4,m_1,m_2,m_3,\mathbf{w_4}$\\
			\hline
	\end{tabular} }	
	{}	
\end{table}	

\begin{table}[ht]	
	\centering\caption	
	{Stable Matchings with Respect to $\SP$. \label{tab:TruncMar}}	
	{\begin{tabular}{|r | l |}	
			\hline
			$\mu_1$&$ \mu_1(m_1)=w_1, \ \mu_1(m_2)=w_2, \ \mu_1(m_3)=w_3, \ \mu_1(m_4)=w_4$\\
			$\mu_2$&$ \mu_2(m_1)=w_2, \ \mu_2(m_2)=w_3, \ \mu_2(m_3)=w_1, \ \mu_2(m_4)=w_4$\\
			$\mu_3$&$ \mu_3(m_1)=w_3, \ \mu_3(m_2)=w_1, \ \mu_3(m_3)=w_2, \ \mu_3(m_4)=w_4$\\
			\hline
	\end{tabular} }	
	{}	
\end{table}

Since we are using the Gale-Shapley algorithm, $r(\SP)=\mu_1$. 
If everyone is honest then the only way that woman $w_1$ can obtain her woman-optimal partner, $m_2$, is by excluding $m_1$ and $m_3$ from her preference list.  
Minimal dishonesty guarantees she will list $m_2$ first.  
Moreover, if she measures honesty with a total ordering then she will correctly order everyone that she indicates she is unwilling to match.  
Thus her preferences must be either $(\PP^1_{w_1}: \ m_2,  m_4, \mathbf{w_1}, m_3, m_1)$ or 
$(\PP^2_{w_1} : \ m_2, \mathbf{w_1}, m_3, m_1, m_4)$. 
$\PP^2_{w_1}$ is the minimally dishonest best response since $K(\SP_{w_1},\PP^1_{w_1})=4$ (the disparities are $\{w_1,m_3\}, \{w_1,m_3\}, \{m_3,m_4\}$, and $\{m_2,m_4\}$) and $K(\SP_{w_1},\PP^2_{w_1}=3)$ ($\{w_1,m_3\}, \{w_1,m_3\}$,  and $\{w_1,m_4\}$).
Moreover, if everyone else is honest it is straightforward to verify that this corresponds to a minimally dishonest equilibrium.  

The equivalent truncated strategies are $(\PP^1_{w_1}: \ m_2, m_4, \mathbf{w_1}, \{m_3, m_1\})$ and $(\PP^2_{w_1} : \ m_2, \mathbf{w_1}, \{m_3,m_1,m_4\})$. 
Similar to before, either $\PP^1_{w_1}$ or $\PP^2_{w_1}$ is her minimally dishonest best response. 
However, $\PP^2_{w_1}$ does not refer to a minimally dishonest best response for $p>\frac{1}{2}$.  
$\PP^1_{w_1}$ is more honest since $K^{(p)}(\SP_{w_1},\PP^1_{w_1})=4+p$ ($D(\SP_{w_1},\PP^1_{w_1})=\{\{w_1, m_1\}, \{w_1,m_3\}, \{m_3,m_4\}, \{m_1,m_4\}\}$ and ${\cal R}(\SP_{w_1},\PP^1_{w_1})=\{m_3, m_1\}$) while $K^{(p)}(\SP_{w_1},\PP^1_{w_1})=3+3p$.  Moreover, it is straightforward to verify that $\PP^1$ is a minimally dishonest equilibrium.  Thus the statement of the proposition holds for $p>\frac{1}{2}$.

We can extend the result for $p>0$, by padding $M$ with $k-4$ additional men that are unwilling to match any woman ($m_i \SP_{m_i} w_j$) for $i> 4$ and for all $j$. 
We extend $\SP_{w_1}$ to ($\SP_{w_1}: m_2,m_3,m_1,m_4,m_5,...,m_k,\mathbf{w_1}$) and have all other women add the additional men to the end of their preference lists. 
Similar to before $w_1$'s minimally dishonest best response is to truncate her list after $m_1$ if she is using a total ordering. 
Now let $\PP_{w_1}$ be her minimally dishonest best response when she uses a truncated list.  
Using minimal dishonesty, it is straightforward to verify that $\PP_{w_1}$ is either ($\PP_{w_1}: \  m_2, m_4, m_5,..., m_{t}, \mathbf{w_1}, \{m_1,m_3,m_{t+1}, m_{t+2},...,m_{k}\}$) for some $t\geq 4$ or ($\PP_{w_1}: \  m_2, \mathbf{w_1}, M\setminus\{m_2\}$). 
Therefore, $K^{(p)}(\SP_{w_1},\PP_{w_1})=k+t-4+p{{k-t+2}\choose{2}}$ where $t\geq 4$ for the former and $t=3$ for the latter.  
For sufficiently large $k$, this is minimized with $t=k$. Thus we can select $k$ so that the minimally dishonest best response is to omit only $m_1$ and $m_3$ completing the proof of the proposition. 
\end{proof}

\section*{Acknowledgments}

Our research has been supported by NSF under grant number
CMMI-1335301. 
The views and conclusions contained in this document are those of
the authors and should not be interpreted as representing the official
policies, either expressed or implied, of the sponsoring organizations,
agencies, or governments.

\bibliographystyle{amsplain}
\bibliography{bib}

\end{document}